\documentclass[lettersize,journal]{IEEEtran}
\IEEEoverridecommandlockouts

\usepackage{amsmath,amsfonts}
\usepackage{algorithmic}
\usepackage{algorithm}
\usepackage{array}
\usepackage[caption=false,font=normalsize,labelfont=sf,textfont=sf]{subfig}
\usepackage{textcomp}
\usepackage{stfloats}
\usepackage{url}
\usepackage{verbatim}
\usepackage{graphicx}
\def\BibTeX{{\rm B\kern-.05em{\sc i\kern-.025em b}\kern-.08em
    T\kern-.1667em\lower.7ex\hbox{E}\kern-.125emX}}
\usepackage{balance}
\usepackage{caption}
\usepackage{siunitx}
\usepackage{tikz}
\usetikzlibrary{shapes.geometric, arrows.meta, positioning, calc, fit}

\usepackage{algorithm}
\usepackage{mathtools}	
\usepackage{amssymb}
\usepackage{amsthm}

\usepackage[utf8]{inputenc}

\usepackage{ifthen}

\usepackage{microtype}

\usepackage{caption}

\usepackage{bm}

\usepackage[draft,nomargin,marginclue,inline]{fixme}
\makeatletter
\renewcommand*\FXLayoutMarginClue[3]{%
  \marginpar[%
  \raggedleft\@fxuseface{margin}\textcolor{red}{\ignorespaces $ \Rightarrow $}]{%
    \raggedright\@fxuseface{margin}\textcolor{red}{\ignorespaces $ \Leftarrow $}}}
\makeatother	

\usepackage{tumcolor}

\usepackage{pgfplotstable}	

\pgfplotsset{
	discard if/.style 2 args={
        x filter/.append code={
            \edef\tempa{\thisrow{#1}}
            \edef\tempb{#2}
            \ifx\tempa\tempb
                
            \fi
        }
    },
    discard if not/.style 2 args={
        x filter/.append code={
            \edef\tempa{\thisrow{#1}}
            \edef\tempb{#2}
            \ifx\tempa\tempb
            \else
                
            \fi
        }
    }
}

\usepackage{cite}

\usepackage[acronym,shortcuts]{glossaries}
\newacronym{cnn}{CNN}{convolutional neural network}
\newacronym{ula}{ULA}{uniform linear array}

\usepackage{cleveref}

\tikzset{algorithm1/.style={mark options={solid},color=TUMBeamerBlue, line width=\lineWidth, mark=square, dashed}}

\pgfplotscreateplotcyclelist{miko}{%
{color=TUMBeamerRed, dash pattern=on 5pt off 1pt on 1pt off 1pt, line width=1.5pt},%
{color=TUMBeamerOrange, dash pattern=on 5pt off 2pt, line width=1.5pt,
mark=o},%
{color=black, dash pattern=on 5pt off 2pt on 3pt off 2pt, line width=1.5pt},%
{color=TUMDarkerBlue, line width=1.5pt},%
{color=black, dotted, line width=1.5pt},%
{color=TUMBeamerGreen, dotted, line width=1.5pt},%
{color=TUMBeamerYellow, dotted, line width=1.5pt},%
{color=TUMBeamerDarkRed, dotted, line width=1.5pt},%
{color=TUMBeamerLightGreen, dotted, line width=1.5pt},%
{color=TUMIvory, dotted, line width=1.5pt},%
{color=TUMLightBlue, dotted, line width=1.5pt},%
}

\DeclareMathOperator*{\argmax}{arg\,max}
\DeclareMathOperator*{\argmin}{arg\,min}
\DeclareMathOperator{\diag}{diag}

\DeclareMathOperator{\expec}{E}

\DeclareMathOperator{\tr}{tr}
\DeclareMathOperator{\rk}{rk}

\DeclareMathOperator{\uvect}{unvec}
\DeclareMathOperator{\vect}{vec}

\newcommand{\calN}{\mathcal{N}}
\newcommand{\calO}{\mathcal{O}}

\newcommand*{\C}{\mathbb{C}}

\newcommand{\herm}{{\operatorname{H}}}
\newcommand{\tp}{{\operatorname{T}}}

\definecolor{myblue}{RGB}{30, 100, 200}

\newlength{\leftstackrelawd}
\newlength{\leftstackrelbwd}
\def\leftstackrel#1#2{\settowidth{\leftstackrelawd}%
	{${{}^{#1}}$}\settowidth{\leftstackrelbwd}{$#2$}%
	\addtolength{\leftstackrelawd}{-\leftstackrelbwd}%
	\leavevmode\ifthenelse{\lengthtest{\leftstackrelawd>0pt}}%
	{\kern-.5\leftstackrelawd}{}\mathrel{\mathop{#2}\limits^{#1}}}
\renewcommand{\stackrel}{\leftstackrel}

\newcommand{\mbA}{\bm{A}}
\newcommand{\mbB}{\bm{B}}
\newcommand{\mbC}{\bm{C}}
\newcommand{\mbD}{\bm{D}}

\newcommand{\mbH}{\bm{H}}

\newcommand{\mbN}{\bm{N}}

\newcommand{\mbP}{\bm{P}}

\newcommand{\mbR}{\bm{R}}
\newcommand{\mbS}{\bm{S}}
\newcommand{\mbT}{\bm{T}}
\newcommand{\mbU}{\bm{U}}
\newcommand{\mbV}{\bm{V}}
\newcommand{\mbW}{\bm{W}}

\newcommand{\mbY}{\bm{Y}}

\newcommand{\mba}{\bm{a}}

\newcommand{\mbd}{\bm{d}}

\newcommand{\mbh}{\bm{h}}

\newcommand{\mbn}{\bm{n}}

\newcommand{\mbv}{\bm{v}}

\newcommand{\mby}{\bm{y}}

\newcommand{\mbbeta}{{\bm{\beta}}}
\newcommand{\mbgamma}{{\bm{\gamma}}}
\newcommand{\mbdelta}{{\bm{\delta}}}

\newcommand{\mbSigma}{{\bm{\Sigma}}}

\newcommand{\mbmu}{{\bm{\mu}}}
\newcommand{\mbkappa}{{\bm{\kappa}}}

\newcommand{\mbzero}{{\bm{0}}}

\newcommand{\inv}{^{-1}}

\newcommand{\hhat}{\hat{\mbh}}

\newcommand{\covhi}{\mbC_i}
\newcommand{\covhk}{\mbC_k}

\usetikzlibrary{intersections,pgfplots.fillbetween,patterns,spy}

\Crefname{figure}{Fig.}{Figs.}
\pgfplotsset{compat=1.15}

\ifdefined\ONECOLUMN

\else

\fi

\definecolor{ourdarkblue}{RGB}{30, 100, 200}
\definecolor{ourdarkgreen}{RGB}{0, 100, 0}
\definecolor{ourdarkorange}{RGB}{201, 98, 18}
\definecolor{ouryellow}{RGB}{220, 210, 50}
\definecolor{myred}{RGB}{255,69,0}
\definecolor{mylila}{RGB}{153,50,204}
\newcommand{\lineWidth}{1.0pt}
\newcommand{\marksize}{2.0pt}
\tikzset{gmmPgmm/.style={mark options={solid},color=TUMBeamerBlue, line width=\lineWidth, mark=square, mark size=\marksize}}
\tikzset{gmmPrnd/.style={mark options={solid},color=TUMBeamerGreen, line width=\lineWidth, mark=o, mark size=\marksize}}
\tikzset{gmmPdft/.style={mark options={solid},color=TUMBeamerOrange, line width=\lineWidth, mark=triangle, mark size=\marksize}}

\tikzset{glmmsePgenie/.style={mark options={solid},color=black, line width=\lineWidth, mark=x, mark size=\marksize}}

\tikzset{dnnPdnn/.style={mark options={solid},color=mylila, line width=\lineWidth, mark=star, mark size=\marksize, dash dot}}

\tikzset{ompPrnd/.style={mark options={solid},color=myred, line width=\lineWidth, mark=diamond, mark size=\marksize, dashed}}
\tikzset{ompPdft/.style={mark options={solid},color=gray, line width=\lineWidth, mark=pentagon, mark size=\marksize, dashed}}

\tikzset{lmmsePrnd/.style={mark options={solid},color=ouryellow, line width=\lineWidth, mark=asterisk, mark size=\marksize, dotted}}
\tikzset{lmmsePdft/.style={mark options={solid},color=brown, line width=\lineWidth, mark=Mercedes star, mark size=\marksize, dotted}}

\newacronym{AWGN}{AWGN}{additive white Gaussian noise}
\newacronym{BLMMSE}{BLMMSE}{Bussgang LMMSE}
\newacronym{BS}{BS}{base station}
\newacronym{CDF}{CDF}{cumulative distribution function}
\newacronym{CNN}{CNN}{convolutional neural network}
\newacronym{CSI}{CSI}{channel state information}
\newacronym{CSIT}{CSIT}{channel state information at the transmitter}
\newacronym{DFT}{DFT}{discrete Fourier transform}
\newacronym{DL}{DL}{downlink}
\newacronym{DNN}{DNN}{deep neural network}
\newacronym{DoA}{DoA}{direction of arrival}
\newacronym{EM}{EM}{expectation maximization}
\newacronym{FDD}{FDD}{frequency division duplex}
\newacronym{GMM}{GMM}{Gaussian mixture model}
\newacronym{LMMSE}{LMMSE}{linear minimum mean square error}
\newacronym{LOS}{LOS}{line of sight}
\newacronym{LS}{LS}{least squares}
\newacronym{MSE}{MSE}{mean squared error}
\newacronym{MIMO}{MIMO}{multiple-input multiple-output}
\newacronym{MPC}{MPC}{multi-path component}
\newacronym{MT}{MT}{mobile terminal}
\newacronym{NLOS}{NLOS}{non-line of sight}
\newacronym{NN}{NN}{neural network}
\newacronym{O2I}{O2I}{outdoor-to-indoor}
\newacronym{OMP}{OMP}{orthogonal matching pursuit}
\newacronym{PDF}{PDF}{probability density function}
\newacronym{PGA}{PGA}{projected gradient ascent}
\newacronym{PSD}{PSD}{power spectral density}
\newacronym{SNR}{SNR}{signal-to-noise ratio}
\newacronym{TDD}{TDD}{time division duplex}
\newacronym{UL}{UL}{uplink}
\newacronym{ULA}{ULA}{uniform linear array}
\newacronym{URA}{URA}{uniform rectangular array}
\newacronym{UMa}{UMa}{urban macrocell}
\newacronym{nSE}{nSE}{normalized spectral efficiency}
\newacronym{cCDF}{cCDF}{complementary cumulative distribution function}
\newacronym{MU-MIMO}{MU-MIMO}{multi-user MIMO}
\newacronym{MU-MISO}{MU-MISO}{multi-user MISO}
\newacronym{BD}{BD}{block diagonalization}
\newacronym{RBD}{RBD}{regularized block diagonalization}
\newacronym{RCI}{RCI}{regularized channel inversion}
\newacronym{WMMSE}{WMMSE}{weighted minimum mean square error}
\newacronym{SWMMSE}{SWMMSE}{stochastic WMMSE}
\newacronym{SVD}{SVD}{singular value decomposition}
\newacronym{SR}{SR}{sum-rate}
\newacronym{CME}{CME}{conditional mean estimator}
\newacronym{ML}{ML}{machine learning}
\newacronym{FLOPS}{FLOPS}{floating-point operations}
\newacronym{OFDM}{OFDM}{orthogonal frequency-division multiplexing}
\newacronym{LTE}{LTE}{Long Term Evolution}
\newacronym{GPS}{GPS}{Global Positioning System}
\newacronym{UMi}{UMi}{urban microcell}
\newacronym{MAP}{MAP}{maximum a posteriori}
\newacronym{3GPP}{3GPP}{3rd Generation Partnership Project}
\newacronym{MMSE}{MMSE}{minimum mean square error}
\newacronym{NMSE}{NMSE}{normalized MSE}
\newacronym{MISO}{MISO}{multiple-input single-output}
\newacronym{CMI}{CMI}{conditional mutual information}
\newacronym{KKT}{KKT}{Karush-Kuhn-Tucker}
\newacronym{RMP}{RMP}{range matching pursuit}

\newcommand{\Nrx}{N_{\mathrm{rx}}}
\newcommand{\Ntx}{N_{\mathrm{tx}}}

\newcommand{\Krx}{K_{\mathrm{rx}}}
\newcommand{\Ktx}{K_{\mathrm{tx}}}

\newtheorem{theorem}{Theorem}
\newtheorem{corollary}{Corollary}[theorem]
\newtheorem{proposition}{Proposition}[theorem]

\begin{document}

\title{A Versatile Pilot Design Scheme for FDD Systems Utilizing Gaussian Mixture Models}

\author{Nurettin~Turan,~\IEEEmembership{Graduate Student Member,~IEEE,} Benedikt~Böck,~\IEEEmembership{Graduate Student Member,~IEEE,} Benedikt~Fesl,~\IEEEmembership{Graduate Student Member,~IEEE,} Michael Joham,~\IEEEmembership{Member,~IEEE,} \\ Deniz~Gündüz,~\IEEEmembership{Fellow,~IEEE,} and Wolfgang~Utschick,~\IEEEmembership{Fellow,~IEEE}\\

\thanks{Preliminary results have been presented at ISWCS'24~\cite{TuFeBoJoUt24}.}
\thanks{The authors acknowledge the financial support by the Federal Ministry of
Education and Research of Germany in the program of ``Souver\"an. Digital.
Vernetzt.''. Joint project 6G-life, project identification number: 16KISK002. 
}
\thanks{N. Turan,  B. Böck, B. Fesl, M. Joham, and W. Utschick are with the TUM School of Computation, Information and
Technology, Technische Universität München, 80333 Munich, Germany
(e-mail: nurettin.turan@tum.de; benedikt.boeck@tum.de; benedikt.fesl@tum.de; joham@tum.de; utschick@tum.de).}
\thanks{D. Gündüz is with the Department of Electrical and Electronic Engineering, Imperial College London, London SW7 2AZ, U.K. (e-mail: d.gunduz@imperial.ac.uk).}
}

\maketitle

\begin{abstract}
In this work, we propose a \ac{GMM}-based pilot design scheme for \ac{DL} channel estimation in single- and multi-user \ac{MIMO} \ac{FDD} systems.
In an initial offline phase, the \ac{GMM} captures prior information during training, which is then utilized for pilot design.
In the single-user case, the \ac{GMM} is utilized to construct a codebook of pilot matrices and, once shared with the \ac{MT}, can be employed to determine a feedback index at the \ac{MT}.
This index selects a pilot matrix from the constructed codebook, eliminating the need for online pilot optimization.
We further establish a sum \ac{CMI}-based pilot optimization framework for \ac{MU-MIMO} systems.
Based on the established framework, we utilize the \ac{GMM} for pilot matrix design in \ac{MU-MIMO} systems.
The analytic representation of the \ac{GMM} enables the adaptation to any \ac{SNR} level and pilot configuration without re-training.
Additionally, an adaption to any number of \acp{MT} is facilitated.
Extensive simulations demonstrate the superior performance of the proposed pilot design scheme compared to state-of-the-art approaches.
The performance gains can be exploited, e.g., to deploy systems with fewer pilots.

\end{abstract}


\section{Introduction}

\begin{tikzpicture}[remember picture,overlay]
\node[anchor=south,yshift=14pt,xshift=-134pt] at (current page.south) {{\parbox{\dimexpr\columnwidth-\fboxsep-\fboxrule\relax}{\footnotesize \quad \copyright This work has been submitted to the IEEE for possible publication. Copyright may be transferred without notice, after which this version may no longer be accessible.}}};
\end{tikzpicture}%

In \ac{MIMO} communication systems, obtaining \ac{CSI} at the \ac{BS} needs to occur in regular time intervals.
In \ac{FDD} systems, both the \ac{BS} and the \acp{MT} transmit at the same time but on different frequencies, which breaks the reciprocity between the instantaneous \ac{UL} \ac{CSI} and \ac{DL} \ac{CSI}.
Consequently, acquiring \ac{DL} \ac{CSI} at the \ac{BS} in \ac{FDD} systems is challenging \cite{2019massive} and, thus, relies on feedback of the estimated channels from the \acp{MT} rendering the quality of \ac{DL} channel estimation critically important.

However, the channel estimation quality is strongly influenced by the choice of pilots, demonstrating the need for sophisticated pilot design schemes in \ac{FDD} systems.
In massive \ac{MIMO} systems, where the \ac{BS} is typically equipped with a high number of antennas, as many pilots as transmit antennas are required to be sent from the \ac{BS} to the \ac{MT} to fully illuminate the channel, i.e., avoiding a systematic error when relying on \ac{LS} \ac{DL} \ac{CSI} estimation at the \ac{MT}.
However, the associated pilot overhead for complete channel illumination can be prohibitive \cite{BjLaMa16}.
In scenarios with spatial correlation at the \ac{BS} and the \acp{MT}, the \ac{DL} training overhead can be significantly reduced by leveraging statistical knowledge of the channel and the noise \cite{ChLoBi14, KoSa04, BjOt10, PaLiZhLu07, FaLiLiGa17, JiMoCaNi15, BaXu17, GuZh19}, e.g., by using Bayesian estimation approaches in combination with well-designed pilot matrices according to different optimization criteria.

In particular, in \cite{ChLoBi14, KoSa04, PaLiZhLu07, BjOt10}, the single-user case is analyzed, where \cite{ChLoBi14} considers a \ac{MISO} system and in \cite{KoSa04, PaLiZhLu07, BjOt10}, a \ac{MIMO} setup with a Kronecker correlation model is investigated.
These works consider the minimization of the \ac{MSE} as the optimization criterion for \ac{DL} channel estimation.
A common approach is to send the eigenvectors of the transmit-side covariance matrix corresponding to the strongest eigenvalues as pilots.
Further, considering different power levels allows to adaptively scale these pilot vectors following a water-filling procedure. 
So-called unitary pilots, i.e., assigning equal power to each pilot vector, can also be used for practical reasons \cite{ChLoBi14}.
In \cite{FaLiLiGa17, BaXu17, JiMoCaNi15, GuZh19}, \ac{MU-MISO} setups, are considered. 
The work \cite{FaLiLiGa17} considers minimizing the sum of \acp{MSE} associated with the channel estimate of all \acp{MT}, whereas \cite{BaXu17} considers minimizing the weighted sum \ac{MSE}. 
Using the close relationship between the mutual information and the \ac{MSE} \cite{DoShVe05}, the work \cite{JiMoCaNi15} considers maximizing the sum \ac{CMI} expression to design the pilot matrix.
In \cite{GuZh19}, the channels are assumed to be distributed according to a Gaussian mixture distribution, and mutual information maximization is adopted.
The works \cite{BaXu17, JiMoCaNi15, GuZh19} apply iterative algorithms for designing the pilot matrices.
However, all of the aforementioned works rely on either perfect or estimated statistical knowledge at the \ac{BS} and/or at the \acp{MT}' side, which may be difficult to acquire.

In this context, the recent works in \cite{MaGa20, MaGü21} propose machine learning-based frameworks for combining the pilot matrix design with \ac{DL} channel estimation.
In particular, end-to-end \acp{DNN} are employed to jointly learn a pilot matrix and a channel estimation module. 
In this way, the designed pilots are optimized for the whole scenario, i.e., a global pilot matrix is learned offline and remains fixed for channel estimation in the online phase.
However, these end-to-end \acp{DNN} cannot provide \ac{MT}-adaptive pilot matrices, neither for the single-user case nor the multi-user case, since knowledge about \ac{MT}-specific statistics, which could be inferred in the online phase, can not be leveraged.
In other words, with the above-mentioned end-to-end \acp{DNN}, pilot matrices tailored for a specific \ac{MT} or a constellation of multiple \acp{MT} are not possible.
Moreover, end-to-end \acp{DNN} are inflexible regarding the number of pilots, since different numbers of pilots require a pilot matrix of suitable dimension and a dedicated channel estimation module, i.e., several end-to-end \acp{DNN} are required. 
The \ac{SNR} level-specific training further hinders practicability.

In recent years, generative models have been utilized for wireless communications entities.
Generally, generative models refer to techniques that aim to learn the underlying distribution of a training data set and provide a prior for different applications.
Generative concepts such as generative adversarial networks \cite{GoPoMiXuWaOzCoBe14}, variational autoencoders \cite{DiMa14}, and diffusion models \cite{HoJaAb20} gained significant attention.
In wireless communications, these generative models have been applied to channel estimation~\cite{BaDoJaDiAn21}, to precoding~\cite{LiLiChLeShLu21}, and as channel modeling frameworks, e.g., \cite{YeLiLiJu20, KiFrSc23}.
In this work, we utilize \acp{GMM}, which are widely adopted in wireless communications research.
For instance, \acp{GMM} are employed for predicting channel states in \cite{MuMiDe20}, for multi-path clustering in \cite{LiZhMaZh20}, and for pilot optimization in \cite{GuZhYi19}.
Note that a fundamental difference of our work to \cite{GuZhYi19} is that we do not assume that the true channel \ac{PDF} is given and is equal to a Gaussian mixture distribution.
More recently, \acp{GMM} were used for channel estimation in \cite{KoFeTuUt21J}, for channel prediction in scenarios with high mobility in \cite{TuBoChFeBuJoFeUt24}, and were leveraged to design a feedback scheme for precoding in single- and multi-user systems in \cite{TuFeKoJoUt23}.
The universal approximation ability of \acp{GMM}, cf. \cite{NgNgChMc20}, justifies the strong performance in the above-mentioned applications.
Beyond this ability, the primary motivation for leveraging \acp{GMM} in this work to propose a versatile pilot design scheme for point-to-point \ac{MIMO} and \ac{MU-MIMO} \ac{FDD} systems is the following.
\acp{GMM} are generative models that comprise a discrete latent space.
This characteristic makes the inference of the latent variable, given an observation, particularly tractable.

The contributions of this work are summarized as follows:
\begin{enumerate}
    \item 
    The proposed \ac{GMM}-based pilot design scheme neither requires \emph{a priori} knowledge of the channel's statistics at the \ac{BS} nor the \ac{MT}.
    The statistical prior information captured by the \ac{GMM} in the initial offline phase is exploited to determine a feedback index at the \ac{MT}, encoded by $B$ bits.
    This feedback index is sufficient to establish common knowledge of the pilot matrix, which is selected from a codebook of pilot matrices.
    Thus, no online pilot optimization is required in the single-user case.

    \item The sum \ac{CMI}-based multi-user pilot optimization framework from \cite{JiMoCaNi15} only considers single antenna \acp{MT}.
    However, with trends towards massive \ac{MIMO}, both the \ac{BS} and the \acp{MT} are equipped with many antennas \cite{AnBuChHaLoSoZh14}.
    In this work, we extend the sum \ac{CMI} maximization to \ac{MU-MIMO} systems, i.e., the \acp{MT} are also equipped with multiple antennas.
    Additionally, we establish a lower bound to the sum \ac{CMI} in a \ac{MU-MIMO} system and discuss conditions when the lower bound is equal to the sum \ac{CMI} or exhibits the largest gap to the sum \ac{CMI}.
    We further recognize that the maximization of the lower bound involves the usage of the iterative algorithm from \cite{JiMoCaNi15} and provides complexity savings in the online pilot matrix optimization.

    \item Based on the extension of the sum \ac{CMI} framework, we utilize the \ac{GMM}-based scheme for pilot matrix design in \ac{MU-MIMO} systems.
    Again, no \emph{a priori} statistical knowledge is required.
    In the multi-user case, each \ac{MT} determines its $B$ bits feedback information and transfers it back to the \ac{BS}.
    Using a feedforward link, the \ac{BS} broadcasts the collected feedback indices to all \acp{MT}.
    The pilot matrix is then found using the respective \ac{GMM}-component’s covariance. 
    Due to the quantized information associated with the component selection of the \ac{GMM}, the proposed scheme offers significant gains over state-of-the-art feedforward signaling schemes, which quantizes the pilot matrix after optimization.

    \item The analytic representation of the \ac{GMM} generally allows the adaption to any \ac{SNR} level and pilot configuration without re-training.
    In the multi-user case, an adaption to any number of \acp{MT} is facilitated.
    Moreover, the \ac{GMM}'s so-called responsibilities that are needed for inferring the feedback index can be further processed to a channel estimate at the \acp{MT}.
    Based on extensive simulations, we highlight the superior performance of the proposed scheme compared to state-of-the-art approaches.
    The gains can be exploited, e.g., to deploy systems with fewer pilots.
\end{enumerate}

\emph{Notation:}
Matrices and vectors are denoted with bold uppercase and bold lowercase letters, respectively.
The transpose, conjugate, and conjugate transpose of a matrix $\mbA$ is denoted by $\mbA^\tp$, $\mbA^*$, and $\mbA^\herm$, respectively.
The all-zeros vector and the identity matrix with appropriate dimensions are denoted by $\mathbf{0}$ or $\mathbf{I}$, respectively.
The canonical unit vector that contains a one at the $i$-th entry and is zero elsewhere is denoted by $\mathbf{e}_i$.
The Euclidean norm of a vector \( \mba \in \C^N \) is denoted by \( \| \mba \| \).
A complex-valued normal distribution with mean vector \( \mbmu \) and covariance matrix \( \mbC \) is denoted by \( \calN_{\C}(\mbmu, \mbC) \) and~$\sim$~stands for ``distributed as.''
The determinant, the rank, and the trace of matrix $\mbA$ is given by $\det(\mbA)$, $\rk(\mbA)$, and $\tr(\mbA)$, respectively.
The vectorization (stacking columns) of a matrix \( \mbA \in \C^{m\times N} \) is written as \( \mba=\vect(\mbA) \in \C^{mN} \), and the reverse operation is denoted by \(\mbA = \uvect_{m,N}(\mba)\).
The vector of the diagonal elements of a diagonal matrix $\mbA$ is denoted by $\diag(\mbA)$, and the diagonal matrix with a vector $\mba$ on its diagonal is denoted by $\diag(\mba)$.
The Kronecker product of two matrices \( \mbA \in \C^{m_1\times N_1} \) and \( \mbB \in \C^{m_2\times N_2} \) is \( \mbA \otimes \mbB \in \C^{m_1m_2\times N_1N_2} \).

\section{System and Channel Model}
\label{sec:system_channel_model}

The system consists of a \ac{BS} equipped with $\Ntx$ antennas and $J$ \acp{MT}.
Each \ac{MT} $j \in \mathcal{J} = \{1,2,\dots,J\}$ is equipped with $\Nrx$ antennas.
We assume a block-fading model, cf., e.g., \cite{NeWiUt18}, where the \ac{DL} signal for block $t$ received at each \ac{MT} $j \in \mathcal{J}$ is
\begin{equation} \label{eq:noisy_obs}
    \mbY_{j,t} = \mbH_{j,t} \mbP^\tp_t + \mbN_{j,t}
\end{equation}
where $t=0,\dots, T$, with the \ac{MIMO} channel $\mbH_{j,t} \in \C^{\Nrx \times \Ntx}$, the pilot matrix $\mbP_t  \in \C^{n_\mathrm{p} \times \Ntx}$, and the \ac{AWGN} $\mbN_{j,t} = [\mbn^{\prime}_{j,t,1}, \dots, \mbn^{\prime}_{j,t,n_\mathrm{p}}] \in \C^{\Nrx \times n_\mathrm{p}}$ with $\mbn^\prime_{j,t,p} \sim \mathcal{N}_\C(\bm{0}, \sigma_n^2 \mathbf{I}_{\Nrx})$ for $p \in \{1,2, \dots, n_\mathrm{p}\}$, and $n_\mathrm{p}$ is the number of pilots.
We consider systems with reduced pilot overhead, i.e., $n_\mathrm{p} < \Ntx$.
For the subsequent analysis, it is advantageous to vectorize~\eqref{eq:noisy_obs}:
\begin{equation} \label{eq:noisy_obs_vec}
    \mby_{j,t} = (\mbP_t \otimes \mathbf{I}_{\Nrx}) \mbh_{j,t} + \mbn_{j,t}
\end{equation}
where \( \mbh_{j,t} = \vect(\mbH_{j,t}) \), \( \mby_{j,t} = \vect(\mbY_{j,t}) \), \( \mbn_{j,t} = \vect(\mbN_{j,t}) \), and $\mbn_{j,t} \sim \mathcal{N}_\C(\mathbf{0}, \mbSigma)$ with $\mbSigma = \sigma_n^2 \mathbf{I}_{\Nrx n_\mathrm{p}}$.

We adopt the \ac{3GPP} spatial channel model~(see \cite{3GPP_mimo,NeWiUt18}) where channels are modeled conditionally Gaussian, i.e., \( \mbh_{j,t} | \mbdelta_j \sim \calN_\C(\mbzero, \mbC_{\mbdelta_j}) \).
The covariance matrix of each \ac{MT} $ \mbC_{\mbdelta_j}$ is assumed to remain constant over $T+1$ blocks.
The random vectors \( \mbdelta_j \) comprise the main angles of arrival/departure of the multi-path propagation cluster between the \ac{BS} and \ac{MT} $j$.
The main angles of arrival/departure are drawn independently and are uniformly distributed over \( [\frac{-\pi}{2}, \frac{\pi}{2}] \).
At the \ac{BS} as well as the \acp{MT} \acp{ULA} are deployed such that the transmit- and receive-side spatial channel covariance matrices are given by
\begin{equation} \label{eq:channel_cov}
    \mbC_{\mbdelta_j}^{\mathrm{\{rx,tx\}}} = \int_{-\pi}^\pi g^{\mathrm{\{rx,tx\}}}(\theta; \mbdelta_j) \mba^{\mathrm{\{rx,tx\}}}(\theta) \mba^{\mathrm{\{rx,tx\},\herm}}(\theta) \mathrm{d}\theta,
\end{equation}
where $\mba^{\mathrm{\{rx,tx\}}}(\theta) = [1, \operatorname{e}^{j\pi\sin(\theta)}, \dots, 
\operatorname{e}^{j\pi(N_{\mathrm{\{rx,tx\}}}-1)\sin(\theta)}]^\tp$ is the array steering vector for an angle of arrival/departure \( \theta \) and \( g^{\mathrm{\{rx,tx\}}}\) is a Laplacian power density whose standard deviation describes the angular spread $\sigma_\text{AS}^\mathrm{\{rx,tx\}}$ of the propagation cluster at the \ac{BS} ($\sigma_\text{AS}^\mathrm{tx}=\SI{2}{\degree}$) and \ac{MT} $j$ ($\sigma_\text{AS}^\mathrm{rx}=\SI{35}{\degree}$) side~\cite{3GPP_mimo}.
The overall channel covariance matrix per \ac{MT} $j$ is constructed as \( \mbC_{\mbdelta_j} = \mbC_{\mbdelta_j}^{\mathrm{tx}} \otimes \mbC_{\mbdelta_j}^{\mathrm{rx}} \) due to the assumption of independent scattering in the vicinity of transmitter and receiver, see, e.g.,~\cite{KeScPeMoFr02}.
In the case of \acp{MT} equipped with a single antenna, \( \mbC_{\mbdelta_j} \) degenerates to the transmit-side covariance matrix \( \mbC_{\mbdelta_j}^{\mathrm{tx}} \).
Due to the low angular spread at the \ac{BS} side, the covariance matrices exhibit a low numerical rank.
The channel model described above was similarly employed in \cite{BaXu17, JiMoCaNi15}, where $\mbdelta_j$ was assumed to be given for all $j \in \mathcal{J}$.

With
\begin{equation} \label{eq:H_dataset}
     \mathcal{H} = \{ \mbh^{(m)} \}_{m=1}^{M},
\end{equation}
we denote the training data set consisting of $M$ channel samples, following \cite{KoFeTuUt21J, FeTuBoUt23}.
For every channel sample $\mbh^{(m)}$, we first generate random angles, collected in \( \mbdelta^{(m)} \), and then draw the sample as \( \mbh^{(m)} \sim \calN_\C(\mbzero, \mbC_{\mbdelta^{(m)}}) \).
The resulting data set represents a wireless communication environment with unknown \ac{PDF} $f_{\mbh}$. 
Additionally, the channel of any \ac{MT} located anywhere within the \ac{BS}'s coverage area can be interpreted as a realization of a random variable with this \ac{PDF} $f_{\mbh}$.

Alternatively, environment-specific training data can be acquired, for instance, from a measurement campaign \cite{TuFeGrKoUt22, TuFeJoMaSoShXiUt24} or by using a ray-tracing tool~\cite{Al19}. 
The analysis of these different training data sources is out of the scope of this work.

\section{Point-to-Point MIMO System} \label{sec:p2pmimo}

\subsection{Pilot Optimization with Perfect Statistical Knowledge} \label{sec:pilot_opt_stat}

In the case of a point-to-point MIMO system, we drop the index $j$ in \eqref{eq:noisy_obs_vec} for notational convenience, yielding
\begin{equation} \label{eq:noisy_obs_vec_p2p}
    \mby_{t} = (\mbP_t \otimes \mathbf{I}_{\Nrx}) \mbh_{t} + \mbn_{t}.
\end{equation}

Given genie knowledge of $\mbdelta$, which fully characterizes the channel statistically (cf. \eqref{eq:channel_cov}), the observation $\mby_t$ is jointly Gaussian with the channel $\mbh_t$.
Thus, we can compute a genie \ac{LMMSE} channel estimate via~\cite{NeWiUt18}
\begin{align}\label{eq:genie_lmmse}
    &\hhat_{t, \text{gLMMSE}} = \expec[\mbh_t \mid \mby_t, \mbdelta] \\
    &= \mbC_{\mbdelta} (\mbP_t \otimes \mathbf{I}_{\Nrx})^\herm ((\mbP_t \otimes \mathbf{I}_{\Nrx}) \mbC_{\mbdelta} (\mbP_t \otimes \mathbf{I}_{\Nrx})^\herm + \mbSigma)^{-1} \mby_t.
\end{align}

The goal of pilot optimization is to design the pilot matrix $\mbP_t$ such that the \ac{MSE} between $\hhat_{t, \text{gLMMSE}}$ and the actual channel $\mbh_t$ is minimized~\cite{FaLiLiGa17, PaLiZhLu07, ChLoBi14}:
\begin{equation} \label{eq:p2p_opt}
    \mbP_t^\star = \argmin_{\mbP_t} \ \expec[\|\hhat_{t, \text{gLMMSE}} - \mbh_t\|^2]
\end{equation}
where the pilot matrix $\mbP_t$ typically satisfies either a total power constraint as in \cite{FaLiLiGa17, PaLiZhLu07} or an equal power per pilot vector constraint as in \cite{ChLoBi14}.
In this work, we consider the latter case.
For a given $\mbdelta$, the optimal pilot matrix  $\mbP_t^\star$ for every block is the same, i.e., $\mbP^\star_t = \mbP_\text{genie}$ for all $t=0,\dots, T$.
In particular, $\mbP_\text{genie}$ is a sub-unitary matrix~\cite{ChLoBi14}
\begin{equation} \label{eq:p_genie}
    \mbP_\text{genie} = \sqrt{\rho} \mbU_\mbdelta^\herm[{:}n_\mathrm{p},:],
\end{equation}
which is composed of the $n_\mathrm{p}$ dominant eigenvectors of the transmit-side covariance matrix \( \mbC_{\mbdelta}^{\mathrm{tx}} = \mbU_\mbdelta \bm{\Lambda}_\mbdelta \mbU_\mbdelta^\herm \) corresponding to the $n_\mathrm{p}$ largest eigenvalues, where $\rho$ denotes the transmit power per pilot vector.

Note that power loading across pilot vectors generally performs better but requires additional processing.
Additionally, with a sub-unitary pilot design, our proposed scheme yields a codebook consisting of pilot matrices that do not depend on the \ac{SNR}, resolving the burden of saving \ac{SNR} level-specific pilot codebooks, see \Cref{sec:pilot_design_gmm}.

\subsection{GMM-based Pilot Design and Downlink Channel Estimation} \label{sec:gmm_based_scheme}

Any channel $\mbh_t$ of a \ac{MT} located anywhere within the \ac{BS}'s coverage area can be interpreted as a realization of a random variable with \ac{PDF} $f_{\mbh}$ for which, however, no analytical expression is available.
To overcome this issue, we utilize a \ac{GMM} to approximate the \ac{PDF} $f_{\mbh}$, following \cite{KoFeTuUt21J, TuFeKoJoUt23}.
This learned model is then shared between the \ac{BS} and the \ac{MT} to establish common awareness of the channel characteristics.
The \ac{GMM} is afterward used to infer feedback information for pilot matrix design and for \ac{DL} channel estimation at the \ac{MT} in the online phase. 
Thereby, the feedback information of the \ac{MT} of a preceding fading block $t-1$ is leveraged at the \ac{BS} to select the pilot matrix for the subsequent fading block $t>0$.

\subsubsection{Modeling the Channel Characteristics at the BS -- Offline} \label{sec:modeling_gmm}

The channel characteristics are captured offline using a \ac{GMM} comprised of $K=2^B$ components,
\begin{equation}\label{eq:gmm_of_h}
    f^{(K)}_{\mbh}(\mbh_t) = \sum\nolimits_{k=1}^K \pi_k \calN_{\C}(\mbh_t; \mbmu_k, \mbC_k)
\end{equation}
where each component of the \ac{GMM} is defined by the mixing coefficient $\pi_k$, the mean $\mbmu_k$, and the covariance matrix $\mbC_k$.

Motivated by the observation that the channel exhibits an unconditioned zero mean and similar to \cite{FeTuBoUt23}, we enforce the means of the \ac{GMM}-components to zero, i.e., $\mbmu_k=\bm{0}$ for all $k \in \{1,\dots, K\}$.
In \cite{BoBaTuSe24}, a theoretical justification for this regularization is provided.
This regularization also reduces the number of learnable parameters and, thus, prevents overfitting.
Note that the parameters of the \ac{GMM}, i.e., $\{\pi_k, \mbC_k\}_{k=1}^K$, remain constant across all blocks.
To obtain maximum likelihood estimates of the \ac{GMM} parameters, we utilize the training data set \(\mathcal{H} \) [see \eqref{eq:H_dataset}] and employ an \ac{EM} algorithm, as described in~\cite[Subsec.~9.2.2]{bookBi06}, where we enforce the means to zero in every M-step of the \ac{EM} algorithm.

For \ac{MIMO} channels, we further impose a Kronecker factorization on the covariances of the \ac{GMM}, i.e., \( \covhk = \mbC^{\mathrm{tx}}_k \otimes \mbC^{\mathrm{rx}}_k \).
Thus, instead of fitting an unconstrained \ac{GMM} with $N\times N$-dimensional covariances (where $N=\Ntx\Nrx$), we fit separate \acp{GMM} for the transmit and receive sides. 
These transmit-side and receive-side \acp{GMM} possess $\Ntx \times \Ntx$ and $\Nrx \times \Nrx$-dimensional covariances, respectively, with $\Ktx$ and $\Krx$ components. 
Then, by computing the Kronecker products of the corresponding transmit- and receive-side covariance matrices, we obtain a \ac{GMM} with \( K = \Ktx\Krx \) components and $N\times N$-dimensional covariances.
Imposing this constraint on the \ac{GMM} covariances significantly decreases the duration of offline training, facilitates parallelization of the fitting process, and demands fewer training samples due to the reduced number of parameters to be learned, cf. \cite{TuFeKoJoUt23, KoFeTuUt21J}.
In addition, having access to a transmit-side covariance during pilot design in the online phase, as discussed in \Cref{sec:pilot_design_gmm}, is ensured.

Using a \ac{GMM}, we can calculate the posterior probability that the channel $\mbh_t$ stems from component $k$ as~\cite[Sec.~9.2]{bookBi06},
\begin{equation}\label{eq:responsibilities_h}
    p(k \mid \mbh_t) = \frac{\pi_k \calN_{\C}(\mbh_t; \bm{0}, \mbC_k)}{\sum_{i=1}^K \pi_i \calN_{\C}(\mbh_t; \bm{0}, \mbC_i) }.
\end{equation}
These posterior probabilities are commonly referred to as responsibilities.


\subsubsection{Sharing the Model with the MTs -- Offline}

A \ac{MT} requires access to the \ac{GMM} parameters to infer the feedback information.
Conceptually, this involves sharing the parameters $\{\pi_k, \mbC_k\}_{k=1}^K$, with the \acp{MT} upon entering the coverage area of the \ac{BS}.
This transfer is required only once since the \ac{GMM} remains unchanged for a specific \ac{BS} environment.

Incorporating model-based insights to restrict the \ac{GMM} covariances, as discussed in \Cref{sec:modeling_gmm}, additionally significantly reduces the model transfer overhead.
Due to specific antenna array geometries, the \ac{GMM} covariances can be further constrained to a Toeplitz or block-Toeplitz matrix with Toeplitz blocks, in case of a \ac{ULA} or \ac{URA}, respectively, cf.~\cite{FeJoHuKoTuUt22, TuFeUt23, BoBaTuSe24}, with even fewer parameters.
In \cite{TuFeUt23}, it is further discussed how \acp{GMM} with variable bit lengths can be obtained.
However, these further structural constraints and the analysis with variable bit lengths are out of the scope of this work.

\subsubsection{Inferring the Feedback Information and Estimating the Channel at the MTs -- Online}

In the online phase, the \ac{MT} infers feedback information given the observation $\mby_t$ utilizing the \ac{GMM}.
The joint Gaussian nature of each \ac{GMM} component [see \eqref{eq:gmm_of_h}] combined with the \ac{AWGN}, allows for simple computation of the \ac{GMM} of the observations with the \ac{GMM} from \eqref{eq:gmm_of_h} as
\begin{equation}\label{eq:gmm_y}
\small
    f_{\mby}^{(K)}(\mby_t) = \sum\nolimits_{k=1}^K \pi_k \calN_{\C}(\mby_t; \bm{0}, (\mbP_t \otimes \mathbf{I}_{\Nrx}) \covhk (\mbP_t \otimes \mathbf{I}_{\Nrx})^\herm + \mbSigma).
\end{equation}

Thus, the \ac{MT} can compute the responsibilities based on the observations $\mby_t$ as 
\begin{equation}\label{eq:responsibilities}
\small
    p(k \mid \mby_t) = \frac{\pi_k \calN_{\C}(\mby_t;\bm{0}, (\mbP_t \otimes \mathbf{I}_{\Nrx}) \covhk (\mbP_t \otimes \mathbf{I}_{\Nrx})^\herm + \mbSigma)}{\sum_{i=1}^K \pi_i \calN_{\C}(\mby_t; \bm{0}, (\mbP_t \otimes \mathbf{I}_{\Nrx}) \covhi (\mbP_t \otimes \mathbf{I}_{\Nrx})^\herm + \mbSigma) }.
\end{equation}
The feedback information $k_t^\star$ is then determined through a \ac{MAP} estimation (cf. \cite{TuFeKoJoUt23}) as
\begin{equation} \label{eq:ecsi_index_j}
    k^\star_t = \argmax_{k } ~{p(k \mid \mby_t)}
\end{equation}
where the index of the component with the highest responsibility for the observation $\mby_t$ serves as the corresponding feedback information.
Hence, the feedback information is simply the index of the \ac{GMM} component that best explains the underlying channel $\mbh_t$ for a given observation $\mby_t$.
Subsequently, the responsibilities are utilized to obtain a channel estimate via the \ac{GMM} by calculating a convex combination of per-component \ac{LMMSE} estimates, as discussed in \cite{TuFeKoJoUt23, KoFeTuUt21J}.
In particular, the \ac{MT} estimates the channel by computing
\begin{equation}\label{eq:gmm_estimator_closed_form}
    \hhat_{t,\text{GMM}}(\mby_t) = \sum\nolimits_{k=1}^K p(k \mid \mby_t) \hhat_{t,\text{LMMSE},k}(\mby_t),
\end{equation}
using the responsibilities $p(k \mid \mby_t)$ from \eqref{eq:responsibilities} and
\begin{align}\label{eq:lmmse_formula}
\small
    &\hhat_{t,\text{LMMSE},k}(\mby_t) \nonumber\\
    &= \covhk (\mbP_t \otimes \mathbf{I}_{\Nrx})^\herm ((\mbP_t \otimes \mathbf{I}_{\Nrx}) \covhk (\mbP_t \otimes \mathbf{I}_{\Nrx})^\herm + \mbSigma)^{-1}\mby_t .
\end{align}
Note that parallelization concerning the number of components $K$ is possible during the feedback inference and the application of the \ac{LMMSE} filters.
Although the \ac{GMM} estimator is a natural choice due to the usage of the responsibilities, which are anyway computed to determine the feedback information, in principle, any other channel estimator can be used.


\subsubsection{Designing the Pilots at the \ac{BS}  -- Online} \label{sec:pilot_design_gmm}

Consider the eigenvalue decomposition of each of the \ac{GMM}'s transmit-side covariances, i.e., \( \mbC_{k}^{\mathrm{tx}} = \mbU_{k} \bm{\Lambda}_{k} \mbU^\herm_{k} \).
For $t>0$, given the feedback information $k_{t-1}^\star$ of the \ac{MT} from the preceding block $t-1$, we propose employing the pilot matrix $\mbP_t$ at the \ac{BS} for the subsequent block $t$ as (cf. \eqref{eq:p_genie})
\begin{equation}
    \mbP_t = \sqrt{\rho} \mbU^\herm_{k^\star_{t-1}}[{:}n_\mathrm{p},:],
\end{equation}
i.e., the $n_\mathrm{p}$ dominant eigenvectors of the $k^\star_{t-1}$-th transmit-side covariance matrix \( \mbC_{k^\star_{t-1}}^{\mathrm{tx}} \) are selected as the pilot matrix.
Since the \ac{GMM}-covariances remain fixed, we can store a set of pilot matrices $\mathcal{P} = \{\mbU^\herm_{k}[{:}n_\mathrm{p},:]\}_{k=1}^K$, and the online pilot design utilizing the proposed \ac{GMM}-based scheme simplifies to a simple selection task based on the feedback information $k^\star_{t-1}$ from the previous block.
For the initial block $t=0$, we employ a \ac{DFT}-based pilot matrix.

\section{Multi-User MIMO System} \label{sec:mumimo}

\subsection{Pilot Optimization with Perfect Statistical Knowledge} \label{sec:mimimo_opt_perfect_stat}

Two well-established possibilities for designing pilot sequences for multi-user setups with single-antenna \acp{MT} were introduced in \cite{JiMoCaNi15, BaXu17}.
The method proposed in \cite{JiMoCaNi15} utilizes the close relationship between the mutual information and the \ac{MSE} \cite{DoShVe05} to design pilot sequences that maximize the sum \ac{CMI} expression subject to a transmit power budget.
Accordingly, the sum mutual information between each \ac{MT}'s channel $\mbh_{j,t}$ and the observations $\mby_{j,t}$ conditioned on the transmitted pilot matrix $\mbP_t$ is considered.
We refer to \cite{JiMoCaNi15} for systems where the \acp{MT} are equipped with a single antenna.
In the following, we extend the framework from \cite{JiMoCaNi15} to systems where the \acp{MT} are equipped with multiple antennas, i.e., to \ac{MU-MIMO} systems.

The sum \ac{CMI} maximization problem for \ac{MU-MIMO} systems for a given number of pilots $n_\mathrm{p}$ and transmit power budget $\rho$ is expressed as
\begin{align} \label{eq:sumcmi_maximization}
    & \max_{\mbP_t} ~\sum_{j=1}^J \log\det \left( \mathbf{I} + \tfrac{1}{\sigma_n^2} (\mbP_t \otimes \mathbf{I}_{\Nrx}) \mbC_{\mbdelta_j} (\mbP_t \otimes \mathbf{I}_{\Nrx})^\herm \right) \\
    & \quad \qquad \text{s.t.} \quad \operatorname{tr}\left(\mbP_t \mbP_t^\herm\right) \leq \rho n_\mathrm{p}. \nonumber
\end{align}
For given $\{\mbdelta_j\}_{j=1}^J$, the optimal pilot matrix  $\mbP_t^\star$ for every block is the same, i.e., $\mbP^\star_t = \mbP_\text{genie}$ for all $t=0,\dots, T$, and we establish the following theorem.
\begin{theorem} \label{thm:optimal_sequnces}
The pilot matrix $\mbP^\star_t$ that maximizes the sum \ac{CMI} in a \ac{MU-MIMO} system provided in \eqref{eq:sumcmi_maximization} satisfies the condition
\begin{equation} \label{eq:sum_cmi_mumimo}
    \sum_{j=1}^J \sum_{i=1}^I \tfrac{1}{\sigma_n^2} \mbV_j \mbD_{\mbbeta_{j,i}} \mbV_j^\herm \mbP^\star_t \mbC_{\mbdelta_j}^{\mathrm{tx}} \tr\left( \mbD_{\mbgamma_{j,i}} \mbT_j\right) = \lambda \mbP^\star_t,
\end{equation}
with the \acp{SVD} for all $j \in \mathcal{J}$,
\begin{align}
    \mbP^\star_t \mbC_{\mbdelta_j}^{\mathrm{tx}} \mbP_t^{\star,\herm} &= \mbV_j \mbS_j \mbV_j^\herm, \label{eq:svd_pctx}\\
    \mbC_{\mbdelta_j}^{\mathrm{rx}} &= \mbW_j \mbT_j \mbW_j^\herm. \label{eq:svd_crx}
\end{align}
Further, let $\sum_{i=1}^{I} \alpha_{j,i} \mbgamma_{j,i} \mbbeta_{j,i}^\tp$ with $I = \min(n_\mathrm{p}, \Nrx)$, be the \ac{SVD} of the matrix $ \uvect_{\Nrx, n_\mathrm{p}} \left( \diag \left( (\mathbf{I} + \tfrac{1}{\sigma_n^2} \mbS_j \otimes \mbT_j )\inv \right) \right)$.
Then, the diagonal matrices $\mbD_{\mbbeta_{j,i}}\in \mathbb{R}^{n_\mathrm{p} \times n_\mathrm{p}}$ and $\mbD_{\mbgamma_{j,i}} \in \mathbb{R}^{\Nrx \times \Nrx}$ for all $j \in \mathcal{J}$, in \eqref{eq:sum_cmi_mumimo}, are defined as 
\begin{align}
    \mbD_{\mbbeta_{j,i}} &= \alpha_{j,i} \diag(\mbbeta_{j,i}), \\
    \mbD_{\mbgamma_{j,i}} &= \diag(\mbgamma_{j,i}),
\end{align}
and $\lambda \geq 0$ is a constant chosen to satisfy the power constraint.

\end{theorem}
\begin{proof}
See \Cref{sec:proof_optimal_sequnces}.
\end{proof}

In the following corollary, we consider the special case with $n_p \leq \Nrx$, which allows for a simplification.
\begin{corollary}
    The pilot matrix $\mbP^\star_t$ that maximizes the sum \ac{CMI} in a \ac{MU-MIMO} system provided in \eqref{eq:sumcmi_maximization} with $n_\mathrm{p} \leq \Nrx$ satisfies the condition
\begin{equation} \label{eq:sum_cmi_mumimo_npleqnrx}
    \sum_{j=1}^J \sum_{i=1}^{n_\mathrm{p}} \tfrac{1}{\sigma_n^2} \mbv_{j,i} \mbv_{j,i}^\herm \mbP^\star_t \mbC_{\mbdelta_j}^{\mathrm{tx}} \tr\left( \mbD_{j,\bar{i}} \mbT_j\right) = \lambda \mbP^\star_t,
\end{equation}
utilizing the \acp{SVD} for all $j \in \mathcal{J}$ from \eqref{eq:svd_pctx} and \eqref{eq:svd_crx}, where $\mbv_{j,i}$ denotes the $i$-th column of $\mbV_j$, and
\begin{align}
    & (\mathbf{I} + \tfrac{1}{\sigma_n^2} \mbS_j \otimes \mbT_j )\inv \nonumber \\
    & = \diag\left([d_{j,1}, \dots, d_{j,n_\mathrm{p}\Nrx} ]^\tp \right) = \sum_{i=1}^{n_\mathrm{p}} \diag(\mathbf{e}_i) \otimes \mbD_{j,\bar{i}}, \label{cor:simplifications}
\end{align}
with the $n_\mathrm{p}$-dimensional vectors $\mathbf{e}_i$, and $\mbD_{j,\bar{i}} = \diag([d_{j,(i-1)\Nrx+1}, \dots, d_{j,i\Nrx}]^\tp)$, cf. \eqref{eq:kpsvd_thm}, where $\lambda \geq 0$ is a constant chosen to satisfy the power constraint. 
\end{corollary}
Note that in \eqref{cor:simplifications} we used that any diagonal matrix $\mbD \in  \mathbb{R}^{{n_\mathrm{p} \Nrx} \times {n_\mathrm{p} \Nrx}}$ with $n_p \leq \Nrx$ can be decomposed as $\mbD =  \diag\left([d_{1}, \dots, d_{n_\mathrm{p}\Nrx} ]^\tp \right) = \sum_{i=1}^{n_\mathrm{p}} \diag(\mathbf{e}_i) \otimes \mbD_{\bar{i}}$,
with the $n_\mathrm{p}$-dimensional vectors $\mathbf{e}_i$, and the $(\Nrx \times \Nrx)$-dimensional diagonal matrices $\mbD_{\bar{i}} = \diag([d_{(i-1)\Nrx+1}, \dots, d_{i\Nrx}]^\tp)$. 

\begin{algorithm}[t]
\captionsetup{font=small}
\caption{Iterative Pilot Matrix Design for \ac{MU-MIMO} Systems}
\label{alg:iterpdesign}
    \begin{algorithmic}[1] \small
    \STATE For $\ell = 0$, initialize pilot matrix $\mbP_t^{(0)}$,
    and set maximum number of iterations $L_{\max}$.
    Compute SVDs for all $j \in \mathcal{J}$, $\mbC_{\mbdelta_j}^{\mathrm{rx}}=\mbW_j \mbT_j \mbW_j^\herm$.
    Set $\varepsilon$.
    Set $\ell=1$.
    \REPEAT
    \STATE Compute SVDs for all $j \in \mathcal{J}$,
    
    $\mbP_t^{(\ell-1)} \mbC_{\mbdelta_j}^{\mathrm{tx}} \mbP_t^{(\ell-1),\herm} = \mbV_j^{(\ell-1)} \mbS_j^{(\ell-1)} \mbV_j^{(\ell-1),\herm}$ 
    \STATE Compute real-valued SVDs for all $j \in \mathcal{J}$, $I=\min(n_\mathrm{p}, \Nrx)$,
    
    $\sum_{i=1}^{I} \alpha_{j,i}^{(\ell-1)} \mbgamma_{j,i}^{(\ell-1)} \mbbeta_{j,i}^{(\ell-1),\tp}$
    
    $=\uvect_{\Nrx, n_\mathrm{p}} \left( \diag \left( (\mathbf{I} + \tfrac{1}{\sigma_n^2} \mbS_j^{(\ell-1)} \otimes \mbT_j )\inv \right) \right)$
    

    \STATE Set for all $j$ and $i$,
    
    $\mbD_{\mbbeta_{j,i}}^{(\ell-1)} = \alpha_{j,i}^{(\ell-1)} \diag\left(\mbbeta_{j,i}^{(\ell-1)}\right)$
        
    $\mbD_{\mbgamma_{j,i}}^{(\ell-1)} = \diag\left(\mbgamma_{j,i}^{(\ell-1)}\right)$
 
    \STATE Calculate intermediate pilot matrix {$\tilde{\mbP}_t^{(\ell)}$

    $= \sum\limits_{j=1}^J \sum\limits_{i=1}^I \frac{\tr( \mbD_{\mbgamma_{j,i}}^{(\ell-1)} \mbT_j)}{\sigma_n^2} \mbV_j^{(\ell-1)} \mbD_{\mbbeta_{j,i}}^{(\ell-1)} \mbV_j^{(\ell-1),\herm} \mbP_t^{(\ell-1)} \mbC_{\mbdelta_j}^{\mathrm{tx}}$}

    \STATE Apply Normalization to fulfill power constraint,

    $\mbP_t^{(\ell)} = \sqrt{\frac{\rho n_\mathrm{p}}{\tr\left(\tilde{\mbP}_t^{(\ell)}\tilde{\mbP}_t^{(\ell),\herm}\right)}} \tilde{\mbP}_t^{(\ell)}$

    \STATE $\ell \gets \ell+1$

    \UNTIL{Spectral norm $\|\mbP_t^{(\ell)} - \mbP_t^{(\ell-1)}\|_2 < \varepsilon$ or $\ell \geq L_{\mathrm{max}}$.}
    \end{algorithmic}
\end{algorithm}

In general, the optimization problem in \eqref{eq:sumcmi_maximization} is a non-convex problem (see the discussion at the end of this subsection).
Thus, similar to the algorithm for \acp{MT} with single antennas~\cite{JiMoCaNi15}, we outline an iterative algorithm that aims to find the optimal pilot matrix based on the first-order \ac{KKT} condition for \ac{MU-MIMO} systems in \Cref{alg:iterpdesign}.
Throughout this work, we have $\varepsilon=10^{-3}$.
As initialization, we apply either a random pilot matrix or a random selection of columns of an (oversampled) \ac{DFT} matrix for $\mbP_t^{(0)}$. 
We discuss the effect of the initialization as well as the specific selection of the maximum number
of iterations $L_{\mathrm{max}}$ on the algorithm's performance in \Cref{sec:sim_results_mumimo}.

To reduce the computational overhead of optimizing \eqref{eq:sumcmi_maximization}, we establish a lower bound on the sum \ac{CMI} expression serving as an approximate objective.
We further present conditions on when the lower bound exhibits the largest gap to the sum \ac{CMI} and when the lower bound is equal to the sum \ac{CMI}.
Accordingly, it is possible to maximize the established lower bound expression instead of the sum \ac{CMI}.

\begin{theorem} \label{thm:lower_bound}
    The sum \ac{CMI} expression of a \ac{MU-MIMO} system provided as the utility function in \eqref{eq:sumcmi_maximization} can be lower bounded by
    \begin{equation} \label{eq:lower_bound}
        \sum_{j=1}^J \log\det \left( \mathbf{I} + \tfrac{1}{\sigma_n^2} (\mbP_t \tr(\mbC_{\mbdelta_j}^{\mathrm{rx}}) \mbC_{\mbdelta_j}^{\mathrm{tx}} \mbP_t^\herm) \right).
    \end{equation}
\end{theorem}
\begin{proof}
See \Cref{sec:proof_lower_bound}.
\end{proof}

\begin{proposition} \label{prop:lower_bound_equal}
    The lower bound expression from \eqref{eq:lower_bound} is equal to the sum \ac{CMI} in a \ac{MU-MIMO} system provided in \eqref{eq:sumcmi_maximization} for the same but arbitrary $\mbP_t$ if and only if for all $j \in \mathcal{J}$
    \begin{equation}
        \rk(\mbC_{\mbdelta_j}^{\mathrm{rx}}) = 1.
    \end{equation}
\end{proposition}
\begin{proof}
See \Cref{sec:proof_lower_bound_equal}.
\end{proof}

\begin{proposition} \label{prop:lower_bound_largestgap}
    The lower bound expression from \eqref{eq:lower_bound} attains the largest gap to the sum \ac{CMI} in a \ac{MU-MIMO} system provided in \eqref{eq:sumcmi_maximization} for the same but arbitrary $\mbP_t$, with fixed traces $\tau_j=\tr(\mbC_{\mbdelta_j}^{\mathrm{rx}})$, if and only if for all $j \in \mathcal{J}$
    \begin{equation}
        \tfrac{1}{\tau_j} \mbC_{\mbdelta_j}^{\mathrm{rx}} = \tfrac{1}{\Nrx}\mathbf{I}.
    \end{equation}
\end{proposition}
\begin{proof}
See \Cref{sec:proof_lower_bound_largestgap}.
\end{proof}

\begin{corollary} \label{col:lower_bound_opt}
    In conjunction with \cite[Th. 1]{JiMoCaNi15},
    the optimal pilot matrix $\mbP^\star_t$ which maximizes the lower bound expression in~\eqref{eq:lower_bound} satisfies the condition 
    \begin{equation} \small
        \sum_{j=1}^J \tfrac{1}{\sigma_n^2} \left(\mathbf{I} + \tfrac{1}{\sigma_n^2} \mbP^\star_t \tr(\mbC_{\mbdelta_j}^{\mathrm{rx}})\mbC_{\mbdelta_j}^{\mathrm{tx}} \mbP_t^{\star,\herm} \right)\inv \mbP^\star_t \tr(\mbC_{\mbdelta_j}^{\mathrm{rx}}) \mbC_{\mbdelta_j}^{\mathrm{tx}} = \lambda \mbP^\star_t.
    \end{equation}
\end{corollary}

Due to \Cref{prop:lower_bound_equal}, there are cases where the lower bound provided in \Cref{thm:lower_bound} equals to the \ac{MU-MISO} sum \ac{CMI} expression (with additional scaling by $\tr(\mbC_{\mbdelta_j}^{\mathrm{rx}})$ for all $j \in \mathcal{J}$) in \cite{JiMoCaNi15} and consequently \cite[Remark 1]{JiMoCaNi15} can be utilized to show that the \ac{MU-MIMO} sum \ac{CMI} maximization problem in \eqref{eq:sumcmi_maximization} is not a convex problem.

The sum \ac{CMI} expression in \eqref{eq:sumcmi_maximization} can be replaced by the lower bound expression from \eqref{eq:lower_bound}.
Then, the resulting objective resembles the \ac{MU-MISO} sum \ac{CMI} maximization as in \cite{JiMoCaNi15} with the additional scaling of the transmit side covariance matrix $\mbC_{\mbdelta_j}^{\mathrm{tx}}$ with $\tr(\mbC_{\mbdelta_j}^{\mathrm{rx}})$ for all $j \in \mathcal{J}$.
The main advantage of the lower bound maximization is a reduced computational complexity (see \Cref{sec:comp_ana}) compared to solving the original problem in \eqref{eq:sumcmi_maximization}, since as a consequence to \Cref{thm:lower_bound} and \Cref{col:lower_bound_opt}, the \ac{MU-MISO} optimizer from~\cite{JiMoCaNi15} can be used for pilot design.

\subsection{Conventional Multi-User Pilot Matrix Signaling} \label{sec:conventional_signaling}

In general, each of the \acp{MT} is unaware of the statistics of the other \acp{MT}.
Thus, after designing the pilot matrix at the \ac{BS} following either the iterative \Cref{alg:iterpdesign} outlined in \Cref{sec:mimimo_opt_perfect_stat} utilizing \Cref{thm:optimal_sequnces} or using the optimizer from \cite{JiMoCaNi15} due to \Cref{thm:lower_bound}, the pilot matrix needs to be transferred from the \ac{BS} to the \acp{MT}.
Transferring the complete pilot matrix to the \acp{MT} might result in a significant signaling overhead that contradicts the initial goal of training overhead reduction~\cite{BaXu17, BaXu18}.
A conventional state-of-the-art solution that avoids transferring the complete pilot matrix was introduced in \cite{BaXu17}.
There, a codebook $\mathcal{C}_{\text{RMP}}$, with $|\mathcal{C}_{\text{RMP}}|=2^{B_{\text{RMP}}}$, known to the \ac{BS} and all \acp{MT} is used to quantize the $n_\mathrm{p}$ pilot vectors of the pilot matrix. 
Subsequently, the selected codebook elements indices are broadcasted to all \acp{MT} using a limited feedforward link.
The method proposed in~\cite{BaXu17} avoids a brute force search of the codebook elements and instead relies on the so-called \ac{RMP} algorithm \cite[Algorithm 2]{BaXu17}.
Thereby, the span of the resulting pilot matrix, which can be encoded by in total $n_\mathrm{p} B_{\text{RMP}}$ bits, provides the best approximation of the span of the originally designed pilot matrix.
As suggested in~\cite{BaXu17, BaXu18}, the $\Ntx$-dimensional columns of an oversampled \ac{DFT} matrix are selected as the elements of the codebook $\mathcal{C}_{\text{RMP}}$.

\subsection{GMM-based Multi-User Pilot Matrix Design and Signaling} \label{sec:pilot_design_gmm_MU}

In the multi-user case, we can exploit the same \ac{GMM} as in the single-user case and, thus, require no additional training or adjustment. 
We propose that each \ac{MT} $j$ determines its feedback information via a \ac{MAP} estimation, similar to the single-user case, as
\begin{equation} \label{eq:ecsi_index_j_mu}
    k^\star_{j,t} = \argmax_{k } ~{p(k \mid \mby_{j,t})}.
\end{equation}
For $t>0$, given the feedback information $\{k_{j,t-1}^\star\}_{j=1}^J$ of all \acp{MT} from the preceding block $t-1$, we propose designing the pilot matrix $\mbP_t$ at the \ac{BS} for the subsequent block $t$ by utilizing the respective \ac{GMM}-component's covariance information (cf. \Cref{sec:modeling_gmm})
\begin{equation} \label{eq:ecsi_index_j_mu_cov}
    \{\mbC_{k_{j,t-1}^\star}=\mbC_{k_{j,t-1}^\star}^{\mathrm{tx}} \otimes \mbC_{k_{j,t-1}^\star}^{\mathrm{rx}}\}_{j=1}^J
\end{equation}
associated with $\{k_{j,t-1}^\star\}_{j=1}^J$.
At the \ac{BS}, the covariance information obtained with the help of the \ac{GMM} can then be leveraged to design the pilot matrix $\mbP_t$ by either applying the iterative algorithm (\Cref{alg:iterpdesign}) outlined in \Cref{sec:mimimo_opt_perfect_stat} following \Cref{thm:optimal_sequnces} or using the iterative algorithm from \cite{JiMoCaNi15} following \Cref{col:lower_bound_opt}.
For the initial block $t=0$, we employ a random selection of columns of an (oversampled) \ac{DFT} matrix as the pilot matrix.
Also, in this case, the \acp{MT} are in general unaware of the statistics of the other \acp{MT}, or in this particular case, the covariance information associated with the \ac{GMM} components. 
In contrast to the conventional feedforward signaling (see \Cref{sec:conventional_signaling}), we propose to broadcast the feedback information $\{k_{j,t-1}^\star\}_{j=1}^J$ of all \acp{MT} using a feedforward link to all \acp{MT}, i.e., the \ac{BS} first collects the feedback indices of all \acp{MT} and subsequently feedforwards them to all \acp{MT}.
Accordingly, the feedforward information is encoded by $JB$ bits.
This approach requires each \ac{MT} to design the pilot matrix locally in order to achieve the common pilot matrix knowledge required for the next fading block.
Thereby, the computational effort at the \acp{MT} can be kept low by either utilizing the lower bound optimization approach or restricting the maximum number of iterations $L_{\mathrm{max}}$ of the iterative algorithm outlined in \Cref{sec:mimimo_opt_perfect_stat}.
Furthermore, the initialization of the pilot matrix of the iterative algorithm needs to be based on the same initial pilot matrix.
We use the \ac{DFT} matrix, which is used in the zeroth block or a random pilot matrix as the initial pilot matrix.
The effect of the specific initialization, the selection of $L_{\mathrm{max}}$, and the choice of the optimization approach allows for a complexity-to-performance trade-off and are discussed in \Cref{sec:comp_ana} and \Cref{sec:sim_results_mumimo}.

\section{Discussion on the Versatility of the Proposed Pilot Matrix Design Scheme} \label{sec:discussion_versatility}

\begin{figure}[t]
    \centering
    \includegraphics[scale=0.93]{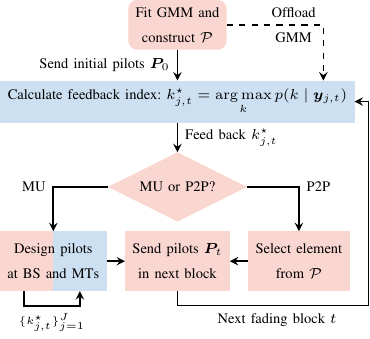}
    \caption{Flowchart of the versatile pilot matrix design scheme. Red (blue) colored nodes are processed at the \ac{BS} (\acp{MT}) and solid (dashed) arrows indicate online (offline) processing.}
    \vspace{-0.4cm}
    \label{fig:flowchart_alt}
\end{figure}

After fitting the \ac{GMM} at the \ac{BS} offline, the model from \eqref{eq:gmm_of_h} is offloaded, i.e., the model parameters are transferred to every \ac{MT} within the coverage area of the \ac{BS}. 
In the online phase, at the \acp{MT}, the \ac{GMM} of the observations [see \eqref{eq:gmm_y}] can be constructed depending on the pilots and \ac{SNR} level and does not require re-training.
By utilizing the \ac{GMM} of the observations, each \ac{MT} can infer its feedback information $k^\star_{j,t}$ by evaluating \eqref{eq:ecsi_index_j_mu}, which is subsequently fed back to the \ac{BS}.
The \ac{BS} can then decide to either serve one \ac{MT} in the point-to-point mode or all \acp{MT} simultaneously in the multi-user mode.
In the single-user mode, the \ac{BS} simply selects the pilot matrix associated with the \ac{MT}'s feedback information of the preceding block $k^\star_{j,t-1}$ from the pilot codebook $\mathcal{P}$, see \Cref{sec:pilot_design_gmm}, and no further processing is required.
In the multi-user mode, a feedforward link is required to broadcast the feedback information $\{k_{j,t-1}^\star\}_{j=1}^J$ of all \acp{MT} from the \ac{BS} to all \acp{MT}.
Then, the pilot matrix $\mbP_t$ for the subsequent fading block can be obtained at the \ac{BS} and \acp{MT} utilizing either the iterative algorithm (\Cref{alg:iterpdesign}) outlined in \Cref{sec:mimimo_opt_perfect_stat} following \Cref{thm:optimal_sequnces} or using the iterative algorithm from \cite{JiMoCaNi15} following \Cref{col:lower_bound_opt}.
The flexibility regarding the transmission mode, together with the possibility to serve any desired number of \acp{MT} $J$, in combination with the straightforward adaption to any number of pilots $n_\mathrm{p}$ and \ac{SNR} level, highlights the versatility of the proposed \ac{GMM}-based pilot design scheme.
The proposed versatile pilot design scheme is summarized as a flowchart in \Cref{fig:flowchart_alt}, where red (blue) colored nodes represent processing steps that are performed at the \ac{BS} (\acp{MT}).


\section{Complexity Analysis} \label{sec:comp_ana}

The online computational complexity of the proposed \ac{GMM}-based scheme can be divided into three parts: 

\subsubsection{Inference of the Feedback Information} Matrix-vector multiplications dominate the computational complexity for inferring the feedback information at the \acp{MT}.
This is because the computation of the responsibilities in~\eqref{eq:responsibilities} entails evaluating Gaussian densities, and the calculations involving determinants and inverses can be pre-computed for a specific \ac{SNR} level due to the fixed \ac{GMM} parameters.
Thus, the inference of the feedback information using \eqref{eq:ecsi_index_j_mu} in the online phase at each \ac{MT} has a complexity of \( \calO(K \Nrx^2 n_\mathrm{p}^2) \), where parallelization concerning the number of components $K$ is possible. 

\subsubsection{Pilot Matrix Design}
The computational complexity of the online pilot matrix design depends on the transmission mode, i.e., single- or multi-user mode.

In the single-user mode, the computational complexity is $\calO(K)$, as it only involves traversing the pre-computed set of pilot matrices $\mathcal{P}$.
Thus, in the online phase, our scheme avoids computing an eigenvalue decomposition, which is required for solving the optimization problem from \eqref{eq:p2p_opt}.

The computational complexity in the multi-user mode depends on the applied optimization algorithm.
In case of the iterative algorithm outlined in \Cref{sec:mimimo_opt_perfect_stat} following \Cref{thm:optimal_sequnces}, the computational complexity per iteration is $\calO(J(n_\mathrm{p}\Ntx^2+n^2_\mathrm{p}\Ntx \min(n_\mathrm{p},\Nrx)))$.
Alternatively, using the iterative algorithm from \cite{JiMoCaNi15} following \Cref{col:lower_bound_opt} exhibits a complexity of $\calO(J n_\mathrm{p}\Ntx^2)$ per iteration.

\subsubsection{Channel Estimation}
Processing the responsibilities to an estimated channel with the \ac{GMM} via \eqref{eq:gmm_estimator_closed_form} exhibits a computational complexity of \(\calO(K\Nrx^2\Ntx n_\mathrm{p})\), since also the \ac{LMMSE} filters for a given \ac{SNR} level can be pre-computed. 
Similar to the feedback inference, parallelization concerning the number of components $K$ is possible when applying the \ac{LMMSE} filters.
As discussed before, in general, any other channel estimator can be used.

\section{Baseline Estimators and Pilot Schemes} \label{sec:baseline_channel_estimators}

Since channel estimation is performed independently at each \ac{MT}, we discuss the estimators from a single \ac{MT} perspective and omit the index $j$ for simplicity.
In addition to the utopian genie \ac{LMMSE} approach \eqref{eq:genie_lmmse}, where we assume perfect knowledge of $\mbdelta$ at the \ac{BS} to design the optimal pilots and at the \ac{MT} to apply the genie \ac{LMMSE} estimator, we consider the following channel estimators and pilot matrices.

Firstly, we consider the \ac{LMMSE} estimator $\hhat_{\text{LMMSE}}$, where the sample covariance matrix is formed using the set $\mathcal{H}$ [see \eqref{eq:H_dataset}], as discussed in, e.g., \cite{TuFeKoJoUt23, KoFeTuUt21J}.

Secondly, we consider a compressive sensing estimation method $\hhat_{\text{OMP}}$ employing the \ac{OMP} algorithm, cf.~\cite{Gharavi, AlLeHe15}.
Since the sparsity order is unknown, but the algorithm's performance heavily depends on it, we use a genie-aided approach to obtain a bound on the algorithm's performance. Specifically, we use perfect channel knowledge to select the optimal sparsity order, cf.~\cite{TuFeKoJoUt23, KoFeTuUt21J}.

Additionally, we compare to an end-to-end \ac{DNN} approach for \ac{DL} channel estimation with a jointly learned pilot matrix $\mbP_\text{DNN}$, similar to \cite{MaGa20, MaGü21}.
To determine the hyperparameters of the \ac{DNN}, we utilize random search \cite{BeBe12}, with the \ac{MSE} serving as the loss function.
The \ac{DNN} architecture comprises $D_\text{CM}$ convolutional modules, each consisting of a convolutional layer, batch normalization, and an activation function, where $D_\text{CM}$ is randomly selected within $[3, 9]$.
Each convolutional layer contains $D_\text{K}$ kernels, where $D_\text{K}$ is randomly selected within $[32, 64]$.
The activation function in each convolutional module is the same and is randomly selected from $ \{\text{ReLu, sigmoid, PReLU, Leaky ReLU, tanh, swish}\}$.
Following a subsequent two-dimensional max-pooling, the features are flattened, and a fully connected layer is employed with an output dimension of~$2\Ntx\Nrx$ (concatenated real and imaginary parts of the estimated channel).
We train a distinct \ac{DNN} for each pilot configuration and \ac{SNR} level, running $50$ random searches per pilot configuration and \ac{SNR} level and selecting the best-performing \ac{DNN} for each setup.

Lastly, as further baseline pilot matrices, we utilize a \ac{DFT} sub-matrix $\mbP_{\text{DFT}}$ as the pilot matrix, see, e.g., \cite{TsZhWa18}, and alternatively, we consider random pilot matrices denoted by $\mbP_{\text{RND}}$, see, e.g.,~\cite{FaLiLiGa17}.

\section{Simulation Results} \label{sec:sim_results}

\begin{figure}[tb]
    \centering
    \includegraphics[]{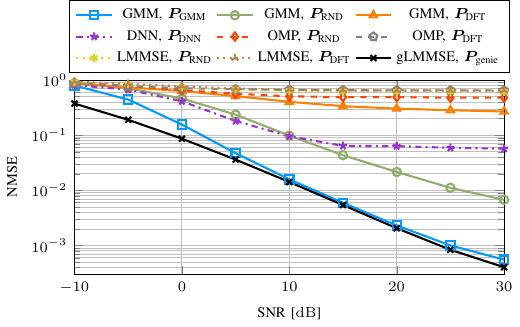}
    \vspace{-0.5cm}
    \caption{The NMSE over the \ac{SNR} for a MIMO system ($\Ntx=16$, $\Nrx=4$) with $B=7$ feedback bits and $n_\mathrm{p}=4$ pilots.}
    \label{fig:oversnr_16x4_np4_B7}
    \vspace{-0.4cm}
\end{figure}

\begin{figure}[tb]
    \centering
    \includegraphics[]{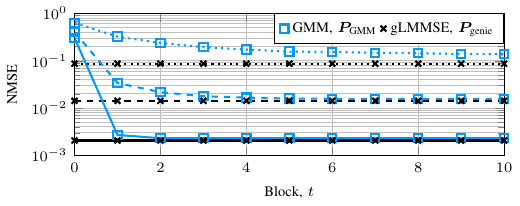}
    \vspace{-0.5cm}
    \caption{The NMSE over the block index $t$ for a MIMO system ($\Ntx=16$, $\Nrx=4$) with $B=7$ feedback bits, $n_\mathrm{p}=4$ pilots, for different $\text{SNR}$ levels (dotted: $\SI{0}{dB}$, dashed: $\SI{10}{dB}$, solid: $\SI{20}{dB}$).}
    \label{fig:overblocks_16x4_np4_B7}
    \vspace{-0.5cm}
\end{figure}

We use the set $\mathcal{H}$ [see \eqref{eq:H_dataset}] with $M=10^5$ samples for fitting the \ac{GMM} and all other data-aided baselines.
We use a different data set of $M_{\text{eval}}=10^4$ channel samples per block $t$ for evaluation purposes, where we set $T=10$.
The data samples are normalized to satisfy \( \expec[\|\mbh\|^2] = N = \Ntx \Nrx \). 
Additionally, we fix $\rho=1$, enabling the definition of the \ac{SNR} as \( \frac{1}{\sigma_n^2} \).
We employ the \ac{NMSE} as the performance measure.
Specifically, for every block $t$, we compute a corresponding channel estimate \( \hhat^{(m_\text{eval})} \)  for each test channel in the set \( \{ \mbh^{(m_\text{eval})} \}_{m_{\text{eval}}=1}^{M_{\text{eval}}} \), and calculate \( \text{NMSE} = \frac{1}{NM_{\text{eval}}} \sum_{m_\text{eval}=1}^{M_\text{eval}} \| \mbh^{(m_\text{eval})} - \hhat^{(m_\text{eval})} \|^2. \)
If not mentioned otherwise, we consider the block with index $t=5$ in the subsequent simulations.

\subsection{Point-to-point MIMO}

In \Cref{fig:oversnr_16x4_np4_B7}, we simulate a system with $\Ntx=16$ \ac{BS} antennas, $\Nrx=4$ \ac{MT} antennas, and $n_\mathrm{p}=4$ pilots. 
Since we have a \ac{MIMO} setup, we fit a Kronecker-structured \ac{GMM} with in total $K=2^7=128$ components ($B=7$ feedback bits), where $\Ktx = 32$ and $\Krx =4$.
The proposed \ac{GMM}-based pilot design scheme (see \Cref{sec:gmm_based_scheme}) denoted by ``GMM, $\mbP_\text{GMM}$'' outperforms all baselines ``\{GMM, OMP, LMMSE\}, \{$\mbP_\text{DFT}$, $\mbP_\text{RND}$\}'' by a large margin, where the \ac{GMM} estimator, the \ac{OMP}-based estimator, or the \ac{LMMSE} estimator, are used in combination with either \ac{DFT}-based pilot matrices or random pilot matrices.
The proposed scheme also outperforms the \ac{DNN} based approach denoted by ``DNN, $\mbP_\text{DNN}$,'' which jointly learns the estimator and a global pilot matrix for the whole scenario; thus, it cannot provide an \ac{MT} adaptive pilot matrix.
This highlights the advantage of the proposed model-based technique over the end-to-end learning technique, which is even trained for each \ac{SNR} level and pilot configuration.
Furthermore, we can observe that the \ac{GMM}-based pilot design scheme performs only slightly worse than the baseline with perfect statistical information at the \ac{BS} and \ac{MT} [see \eqref{eq:genie_lmmse} and \eqref{eq:p_genie}], denoted by  ``gLMMSE, $\mbP_\text{genie}$,'' being a utopian estimation approach.
We can observe a larger gap in the low \ac{SNR} regime, where the feedback information obtained through the responsibilities of a given observation [see \eqref{eq:ecsi_index_j}] is less accurate due to high noise.

In \Cref{fig:overblocks_16x4_np4_B7}, we analyze the performance of the proposed \ac{GMM}-based pilot design scheme over the block index $t$ for the same setup as before at three different \ac{SNR} levels, i.e., $\text{SNR} \in \{\SI{0}{dB}, \SI{10}{dB}, \SI{20}{dB}\}$.
As discussed in \Cref{sec:pilot_design_gmm}, at $t=0$, we utilize \ac{DFT}-based pilots.
At $t=1$, we can already see a significant gain in performance of the proposed \ac{GMM}-based pilot design scheme due to the feedback of the index.
The results further reveal that with an increasing \ac{SNR}, fewer blocks are required to achieve a performance close to the utopian baseline ``gLMMSE, $\mbP_\text{genie}$'' which requires perfect statistical knowledge at the \ac{BS} and the \ac{MT}.

\begin{figure}[tb]
    \centering
    \includegraphics[]{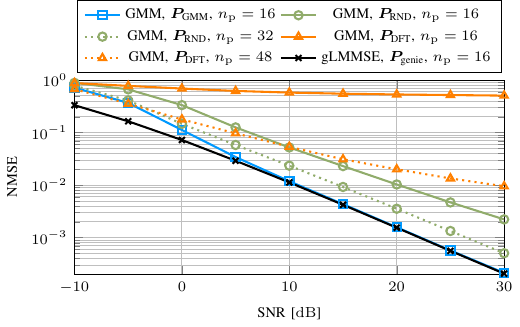}
    \vspace{-0.5cm}
    \caption{The NMSE over the \ac{SNR} for a MISO system ($\Ntx=64$, $\Nrx=1$) with $B=6$ feedback bits and $n_\mathrm{p} \in \{16,32,48\}$ pilots.}
    \label{fig:oversnr_64_npX}
    \vspace{-0.6cm}
\end{figure}

In the remainder, we consider a \ac{MISO} system with $\Ntx=64$ \ac{BS} antennas and $\Nrx=1$ antenna at the \ac{MT}.
In \Cref{fig:oversnr_64_npX} we utilize a \ac{GMM} with $K=2^6=64$ components ($B=6$ feedback bits) and consider setups with $n_\mathrm{p} \in \{16,32,48\}$ pilots.
In this case, the proposed scheme ``GMM, $\mbP_\text{GMM}$, $n_\mathrm{p}=16$,'' performs only slightly worse than the genie-aided approach ``gLMMSE, $\mbP_\text{genie}$.'' 
Moreover, the \ac{GMM}-based pilot design scheme outperforms the baselines ``GMM, \{$\mbP_\text{DFT}$, $\mbP_\text{RND}$\}, $n_\mathrm{p} \in \{16,32,48\}$.''
In particular, the proposed scheme with only $n_\mathrm{p}=16$ pilots outperforms random pilot matrices with twice as much pilots ($n_\mathrm{p}=32$) or in case of \ac{DFT}-based pilots even thrice as much pilots ($n_\mathrm{p}=48$).

Lastly, in \Cref{fig:overcomps_64_np16}, we analyze the effect of varying the number of \ac{GMM}-components $K$, on the scheme's performance, where we set $n_\mathrm{p}=16$ and consider three different \ac{SNR} levels, i.e., $\text{SNR} \in \{\SI{0}{dB}, \SI{10}{dB}, \SI{20}{dB}\}$.
We observe that the estimation error decreases as the number of components $K$ increases.
Moreover, as the \ac{SNR} increases, the gap to the genie-aided approach ``gLMMSE, $\mbP_\text{genie}$'' narrows.
Above $K=32$ components, a saturation can be observed.
These results suggest that varying the number of \ac{GMM} components $K$, allows for a performance-to-complexity trade-off without sacrificing too much in performance for $K\geq16$ components.

\begin{figure}[tb]
    \centering
    \includegraphics[]{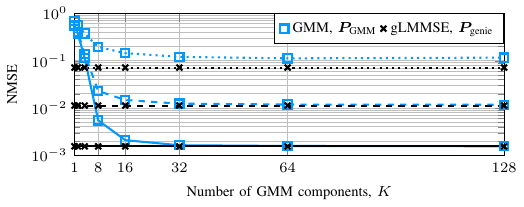}
    \vspace{-0.5cm}
    \caption{The NMSE over the number of components $K=2^B$ for a MISO system ($\Ntx=64$, $\Nrx=1$) with $n_\mathrm{p} = 16$ pilots, for different $\text{SNR}$ levels (dotted: $\SI{0}{dB}$, dashed: $\SI{10}{dB}$, solid: $\SI{20}{dB}$).}
    \label{fig:overcomps_64_np16}
    \vspace{-0.4cm}
\end{figure}

\subsection{Multi-user MIMO} \label{sec:sim_results_mumimo}

In the following multi-user simulations, the performance is measured by the \ac{NMSE} averaged over the number of \acp{MT} $J$ and $M_{\text{con}}=500$ constellations, with $J$ \acp{MT} randomly drawn from our evaluation set for each constellation.
If not mentioned otherwise, we consider the block with index $t=5$ in the simulations, do not specify a maximum number of iterations $L_{\max}$ for the pilot optimization algorithm (\Cref{alg:iterpdesign}) from \Cref{sec:mimimo_opt_perfect_stat}, and initialize with a random pilot matrix.

In \Cref{fig:MU_fig_oversnr_32x4_np8_B7}, we consider a \ac{MU-MIMO} system, with $J=4$ \acp{MT}, $\Ntx=32$ \ac{BS} antennas, $\Nrx=4$ \ac{MT} antennas, and $n_\mathrm{p}=8$ pilots. 
Due to the \ac{MIMO} setup, we fit a Kronecker-structured \ac{GMM} with in total $K=2^7=128$ components ($B=7$ feedback bits per \ac{MT}), where $\Ktx = 32$ and $\Krx =4$.
The proposed \ac{GMM}-based pilot design approach ``GMM, $\mbP_\text{GMM}$'' (see \Cref{sec:pilot_design_gmm_MU}) outperforms the baselines ``GMM, \{$\mbP_\text{DFT}$, $\mbP_\text{RND}$\}'' by a large margin, where the \ac{GMM} estimator is used in combination with either random pilot matrices or a random selection of columns of the two times oversampled \ac{DFT} matrix as the pilot matrix.
The baselines ``gLMMSE, $\{\mbP_\text{DFT}, \mbP_\text{RND}\}$'' assume that each \ac{MT} has perfect statistical knowledge of its channel to apply the utopian genie \ac{LMMSE} approach \eqref{eq:genie_lmmse} for channel estimation in combination with random or \ac{DFT}-based pilots as described above.
The \ac{GMM}-based scheme outperforms the baseline ``gLMMSE, $\mbP_\text{DFT}$'' over the whole \ac{SNR} range, and the baseline ``gLMMSE, $\mbP_\text{RND}$'' for \ac{SNR} values larger than $\SI{0}{dB}$. 
In addition to the assumption of perfect statistical knowledge at each \ac{MT} about its own channel, the baseline ``gLMMSE, $\mbP_\text{genie}$'' further assumes perfect statistical knowledge of the channels of all \acp{MT} at the \ac{BS} in order to design the pilot matrix using \Cref{alg:iterpdesign} from \Cref{sec:mimimo_opt_perfect_stat}.
The \ac{GMM}-based scheme performs only slightly worse than this utopian baseline.
A larger gap is present in the low \ac{SNR} regime, where the feedback information of each \ac{MT} obtained through the responsibilities [see \eqref{eq:ecsi_index_j_mu}] is less accurate due to high noise.

\begin{figure}[tb]
    \centering
    \includegraphics[]{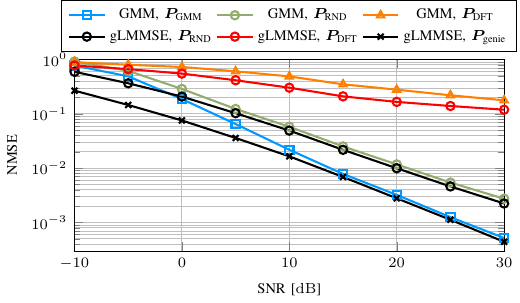}
    \vspace{-0.5cm}
    \caption{The NMSE over the \ac{SNR} for a MU-MIMO system ($\Ntx=32$, $\Nrx=4$) with $J=4$ \acp{MT}, $B=7$ feedback bits and $n_\mathrm{p}=8$ pilots.}
    \label{fig:MU_fig_oversnr_32x4_np8_B7}
    \vspace{-0.5cm}
\end{figure}

\begin{figure}[tb]
    \centering
    \includegraphics[]{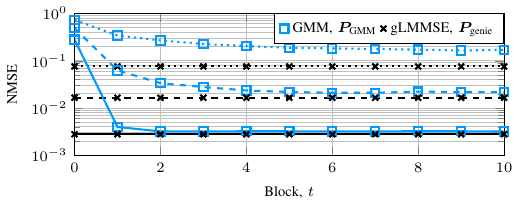}
    \vspace{-0.5cm}
    \caption{The NMSE over the block index $t$ for a MU-MIMO system ($\Ntx=32$, $\Nrx=4$) with $J=4$ \acp{MT}, $B=7$ feedback bits, $n_\mathrm{p}=8$ pilots, for different $\text{SNR}$ levels (dotted: $\SI{0}{dB}$, dashed: $\SI{10}{dB}$, solid: $\SI{20}{dB}$).}
    \label{fig:MU_fig_overblocks_32x4_np8_B7}
    \vspace{-0.5cm}
\end{figure}

In \Cref{fig:MU_fig_overblocks_32x4_np8_B7}, we analyze the performance of the \ac{GMM}-based pilot design scheme over the block index $t$ for the same \ac{MU-MIMO} setup as before at three different \ac{SNR} levels, i.e., $\text{SNR} \in \{\SI{0}{dB}, \SI{10}{dB}, \SI{20}{dB}\}$.
As discussed in \Cref{sec:pilot_design_gmm_MU}, at $t=0$, we utilize a random selection of columns of the two times oversampled \ac{DFT} matrix as the pilot matrix.
Similar to the single-user case (see \Cref{fig:overblocks_16x4_np4_B7}), after only one block, we can see a significant gain in performance of the proposed \ac{GMM}-based pilot design scheme due to the feedback of the index of each \ac{MT}.
Additionally, with an increasing \ac{SNR}, fewer blocks are required to achieve a performance close to the utopian baseline ``gLMMSE, $\mbP_\text{genie}$'' which requires perfect statistical knowledge at the \ac{BS} and the \acp{MT}.


In \Cref{fig:MU_fig_overmaxiter_32x4_np16}, for the same \ac{MU-MIMO} system with $n_\mathrm{p}=16$ pilots, the effect of using the iterative algorithm from \cite{JiMoCaNi15} following \Cref{col:lower_bound_opt} due to the established lower bound in \Cref{thm:lower_bound} instead of the iterative approach (\Cref{alg:iterpdesign}) outlined in \Cref{sec:mimimo_opt_perfect_stat}, subject to a maximum number of iterations $L_{\max}$, is analyzed.
Moreover, in addition to the initialization of the iterative algorithm with a random pilot matrix, we consider a random selection of columns of the two times oversampled \ac{DFT} matrix as an alternative initially chosen pilot matrix.
Accordingly, the curves labeled ``\{GMM, gLMMSE\}, $\mbP_\text{GMM}$, \{LB-RND, LB-DFT\}'' utilize the lower bound maximization, with either random or \ac{DFT}-based initialization using the \ac{GMM}-based feedback scheme, or perfect statistical knowledge at the \ac{BS} and the \acp{MT}, respectively.
With ``gLMMSE, $\mbP_\text{genie}$,'' we denote the baseline where the iterative algorithm from \Cref{sec:mimimo_opt_perfect_stat} is used in combination with random initialization and no maximum number of iterations $L_{\max}$ being specified.
\begin{figure}[tb]
    \centering
    \includegraphics[]{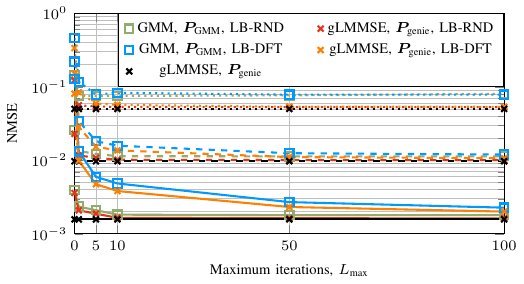}
    \vspace{-0.5cm}
    \caption{The NMSE over the number of maximum iterations $L_{\max}$ for a MU-MIMO system ($\Ntx=32$, $\Nrx=4$) with $J=4$ \acp{MT}, $B=7$ feedback bits, and $n_\mathrm{p} = 16$ pilots for different $\text{SNR}$ levels (dotted: $\SI{0}{dB}$, dashed: $\SI{10}{dB}$, solid: $\SI{20}{dB}$). }
    \label{fig:MU_fig_overmaxiter_32x4_np16}
    \vspace{-0.3cm}
\end{figure}
\begin{figure}[tb]
    \centering
    \includegraphics[]{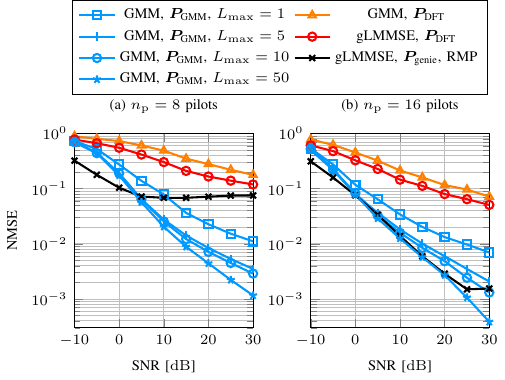}
    \vspace{-0.1cm}
    \caption{The NMSE over the \ac{SNR} for a MU-MIMO system ($\Ntx=32$, $\Nrx=4$) with $J=4$ \acp{MT}, $B=7$ feedback bits, and (a) $n_\mathrm{p}=8$ pilots, or (b) $n_\mathrm{p}=16$ pilots.}
    \label{fig:MU_fig_oversnr_32x4_np8_AND_np16_B7_withRMP}
    \vspace{-0.4cm}
\end{figure}
In \Cref{fig:MU_fig_overmaxiter_32x4_np16}, we can see that for low values for $L_{\max}$, a random initialization of the iterative lower-bound maximization is advantageous as compared to a \ac{DFT}-based initialization, irrespective of whether we use the \ac{GMM}-based scheme or require perfect statistical knowledge.
For a high \ac{SNR} value of $\SI{20}{dB}$ (solid), the \ac{DFT}-based initialization requires larger values for $L_{\max}$ to achieve a similar performance as the random initialization.
However, for the \ac{GMM}-based scheme, where an information feedforward from the \ac{BS} to the \acp{MT} is required, \ac{DFT}-based initialization is advantageous as it eliminates the need for seed exchange and a common random number generator between the \ac{BS} and \acp{MT}, enhancing practicability.
We can see further that for an \ac{SNR} of $\SI{0}{dB}$ (dotted), there is a gap present between the \ac{GMM}-based scheme compared to the case where perfect statistical knowledge is required since the feedback information of each \ac{MT} obtained through the responsibilities given an observation [see \eqref{eq:ecsi_index_j_mu}] is less accurate due to high noise.
We have observed this behavior for different numbers of pilots $n_\mathrm{p}$.
Overall, this analysis suggests that with the lower bound maximization using the iterative algorithm from \cite{JiMoCaNi15} following \Cref{col:lower_bound_opt}, together with a relatively low number of maximum iterations $L_{\max}$, a pilot matrix design of relatively low complexity is possible without sacrificing too much performance.
In the next simulations, we focus our analysis on \ac{DFT}-based initialization due to superior practicability compared to random initialization.

In \Cref{fig:MU_fig_oversnr_32x4_np8_AND_np16_B7_withRMP}(a) and \Cref{fig:MU_fig_oversnr_32x4_np8_AND_np16_B7_withRMP}(b), we again consider the \ac{MU-MIMO} setup from before, with $n_\mathrm{p}=8$ pilots or $n_\mathrm{p}=16$ pilots, respectively. 
The proposed \ac{GMM}-based pilot design approaches labeled ``GMM, $\mbP_\text{GMM}$, $L_{\max} \in$ \{$1,5,10,50$\}'' utilize the lower-bound maximization by using the iterative algorithm from \cite{JiMoCaNi15} following \Cref{col:lower_bound_opt} together with \ac{DFT}-based initialization with different maximum numbers of iterations $L_{\max}$.
We can observe that even one iteration ($L_{\max}=1$) is enough to outperform 
the baselines ``\{GMM, gLMMSE\}, $\mbP_\text{DFT}$,'' by a large margin.
In \Cref{fig:MU_fig_oversnr_32x4_np8_AND_np16_B7_withRMP}(a), for \ac{SNR} values larger than approximately $\SI{5}{dB}$ and even $L_{\max}=5$, the proposed scheme even outperforms ``gLMMSE, $\mbP_\text{genie}$, RMP,'' i.e., the state-of-the-art pilot matrix signaling scheme from \Cref{sec:conventional_signaling} with $B_\text{RMP}=7$ bits; there, each \ac{MT} has perfect statistical knowledge of its channel to apply the utopian genie \ac{LMMSE} approach from \eqref{eq:genie_lmmse} for channel estimation, and perfect statistical knowledge of the channels of all \acp{MT} at the \ac{BS} is assumed to design the pilot matrix using the iterative algorithm from \Cref{sec:mimimo_opt_perfect_stat} prior to quantization with the \ac{RMP} algorithm.
The error floor with the \ac{RMP}-based scheme was already similarly observed and discussed in \cite{BaXu17}.
In \Cref{fig:MU_fig_oversnr_32x4_np8_AND_np16_B7_withRMP}(b), with $n_\mathrm{p}=16$ pilots, the error floor of the \ac{RMP}-based scheme is present for \ac{SNR} values larger than $\SI{25}{dB}$.
The proposed \ac{GMM}-based scheme, which does not require any \emph{a priori} statistics, performs very well also in this setup.
Note that in these particular setups, the \ac{RMP}-based scheme requires in total $n_\mathrm{p} B_{\text{RMP}}$ feedforward bits, yielding $56$ bits for the setup in \Cref{fig:MU_fig_oversnr_32x4_np8_AND_np16_B7_withRMP}(a), and $112$ bits for the setup in \Cref{fig:MU_fig_oversnr_32x4_np8_AND_np16_B7_withRMP}(b), whereas the \ac{GMM}-based scheme does not depend on the number of pilots $n_\mathrm{p}$ and only requires in total $JB=28$ feedforward bits.

\begin{figure}[tb]
    \centering
    \includegraphics[]{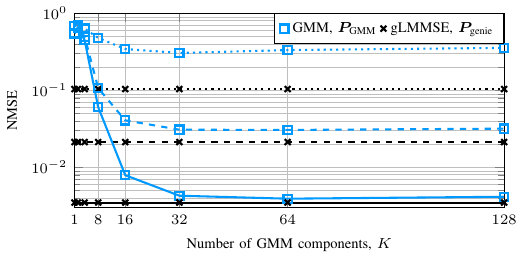}
    \vspace{-0.5cm}
    \caption{The NMSE over the number of components $K=2^B$ for a MU-MISO system ($\Ntx=64$, $\Nrx=1$) with $J=8$ \acp{MT}, $n_\mathrm{p} = 16$ pilots, for different $\text{SNR}$ levels (dotted: $\SI{0}{dB}$, dashed: $\SI{10}{dB}$, solid: $\SI{20}{dB}$).}
    \label{fig:MU_overcomps_64_np16}
    \vspace{-0.5cm}
\end{figure}

Lastly, in \Cref{fig:MU_overcomps_64_np16}, we analyze the effect of varying the number of \ac{GMM}-components $K$ on the scheme's performance in a \ac{MU-MISO} system with $J=8$ single-antenna \acp{MT} and $\Ntx=64$ \ac{BS} antennas.
We set $n_\mathrm{p}=16$ pilots and consider three different \ac{SNR} levels, i.e., $\text{SNR} \in \{\SI{0}{dB}, \SI{10}{dB}, \SI{20}{dB}\}$.
Since we consider single-antenna \acp{MT} in this setup, we apply the iterative algorithm from \cite{JiMoCaNi15}, utilize \ac{DFT}-based initialization, and set the maximum number of iterations $L_{\max}=50$, with the \ac{GMM}-based scheme.
For the ``gLMMSE, $\mbP_\text{genie}$'' baseline, random initialization is applied, and no maximum number of iterations $L_{\max}$ is specified.
Similar to the single-user case in \Cref{fig:overcomps_64_np16}, we can observe that the estimation error decreases with an increasing number of components $K$.
With an increasing \ac{SNR}, the gap to the genie-aided approach ``gLMMSE, $\mbP_\text{genie}$'' decreases.
Compared to the single-user case, slightly more \ac{GMM}-components are needed in a multi-user system.
An intuitive explanation is that with the \ac{GMM}-based scheme, the true but unknown covariance of each \ac{MT} is approximated, and the approximation error accumulates over multiple \acp{MT}. 
With more components, the approximation error decreases, yielding a better performance.
Overall, also in the multi-user case, these results suggest that a performance-to-complexity trade-off can be realized by varying the number of \ac{GMM} components $K$.

\section{Conclusion} \label{sec:conclusion}

In this work, we proposed a \ac{GMM}-based pilot design scheme for single- and multi-user \ac{MIMO} \ac{FDD} systems.
A significant advantage of the proposed scheme is that it does not require \emph{a priori} knowledge of the channel's statistics at the \ac{BS} and the \acp{MT}.
Instead, it relies on a feedback mechanism, establishing common knowledge of the pilot matrix in a single-user system.
In the multi-user case, the pilot design involves an additional feedforward of the indices of the \ac{GMM} components. 
The \ac{GMM}-based scheme offers significant versatility since it can generally be adapted to any desired \ac{SNR} level, pilot configuration, and different numbers of \acp{MT} without requiring re-training.
Simulation results show that the performance gains achieved with the proposed scheme allow the deployment of system setups with reduced pilot overhead while maintaining a similar estimation performance.
In future work, the extension of the \ac{GMM}-based pilot scheme to systems with a spatio-temporal correlation of channels \cite{NoZoSuLo14, BaXu17, BoBaRiUt23} can be investigated.

\appendix
\subsection{Proof of \Cref{thm:optimal_sequnces}}
\label{sec:proof_optimal_sequnces}

We start by rewriting the sum \ac{CMI} expression in \eqref{eq:sumcmi_maximization} by utilizing for all $j \in \mathcal{J}$, \( \mbC_{\mbdelta_j} = \mbC_{\mbdelta_j}^{\mathrm{tx}} \otimes \mbC_{\mbdelta_j}^{\mathrm{rx}} \), and the Kronecker product rules (see \cite{Sc04}) as
\begin{align}
    & \sum_{j=1}^J \log\det \left( \mathbf{I} + \tfrac{1}{\sigma_n^2} (\mbP_t \otimes \mathbf{I}_{\Nrx}) \mbC_{\mbdelta_j} (\mbP_t \otimes \mathbf{I}_{\Nrx})^\herm \right) \nonumber \\
    &=\sum_{j=1}^J \log\det \left( \mathbf{I} + \tfrac{1}{\sigma_n^2} (\mbP_t \mbC_{\mbdelta_j}^{\mathrm{tx}} \mbP_t^\herm)  \otimes \mbC_{\mbdelta_j}^{\mathrm{rx}} \right).  \label{eq:reform_scmi}
\end{align}

The partial derivative of the sum \ac{CMI} with respect to
the entry of the conjugate of the pilot matrix $[\mbP_t^*]_{r,c}$ at row $r$ and column $c$ is then given by

\begin{align}
    &\dfrac{\partial}{\partial[\mbP_t^*]_{r,c}} \left(\sum_{j=1}^J \log\det \left( \mathbf{I} + \tfrac{1}{\sigma_n^2} (\mbP_t \mbC_{\mbdelta_j}^{\mathrm{tx}} \mbP_t^\herm)  \otimes \mbC_{\mbdelta_j}^{\mathrm{rx}} \right) \right) \nonumber \\
    &= \sum_{j=1}^J \tr\left( \left(\mathbf{I} + \tfrac{1}{\sigma_n^2} (\mbP_t \mbC_{\mbdelta_j}^{\mathrm{tx}} \mbP_t^\herm)  \otimes \mbC_{\mbdelta_j}^{\mathrm{rx}} \right)\inv \right. \nonumber \label{eq:linearity_sum_der} \\
    &\qquad \qquad \quad \times \left. \dfrac{\partial}{\partial[\mbP_t^*]_{r,c}} \left( \mathbf{I} + \tfrac{1}{\sigma_n^2} (\mbP_t \mbC_{\mbdelta_j}^{\mathrm{tx}} \mbP_t^\herm)  \otimes \mbC_{\mbdelta_j}^{\mathrm{rx}} \right) \right) \\
    &= \sum_{j=1}^J \tr \bigg( \left(\mathbf{I} + \tfrac{1}{\sigma_n^2} (\mbP_t \mbC_{\mbdelta_j}^{\mathrm{tx}} \mbP_t^\herm)  \otimes \mbC_{\mbdelta_j}^{\mathrm{rx}} \right)\inv \nonumber \\
    &\qquad \qquad \quad \times \left( \tfrac{1}{\sigma_n^2} (\mbP_t \mbC_{\mbdelta_j}^{\mathrm{tx}} \mathbf{e}_c \mathbf{e}_r^\tp)  \otimes \mbC_{\mbdelta_j}^{\mathrm{rx}} \right) \bigg), \label{eq:eldev_thm}
\end{align}
where \eqref{eq:linearity_sum_der} follows from the linearity of the sum, and using the derivation rule provided in \cite[eq. (46)]{PePe08}.
By utilizing the \acp{SVD}
\begin{align}
    & \mbP_t \mbC_{\mbdelta_j}^{\mathrm{tx}} \mbP_t^\herm = \mbV_j \mbS_j \mbV_j^\herm, \label{eq:svd_covtx} \\
    & \mbC_{\mbdelta_j}^{\mathrm{rx}}=\mbW_j \mbT_j \mbW_j^\herm, \label{eq:svd_covrx}
\end{align}
using $\mathbf{I} = (\mbV_j \otimes \mbW_j) (\mbV_j \otimes \mbW_j)^\herm $, and applying basic Kronecker product rules, we have that
\begin{align}
    &\left(\mathbf{I} + \tfrac{1}{\sigma_n^2} (\mbP_t \mbC_{\mbdelta_j}^{\mathrm{tx}} \mbP_t^\herm)  \otimes \mbC_{\mbdelta_j}^{\mathrm{rx}} \right)\inv \nonumber \\
    &= (\mbV_j \otimes \mbW_j) (\mathbf{I} + \tfrac{1}{\sigma_n^2} \mbS_j \otimes \mbT_j)\inv (\mbV_j \otimes \mbW_j)^\herm. \label{eq:svd_inv_diag}
\end{align}
Note that $\mbS_j \otimes \mbT_j$ is diagonal.
Let $\sum_{i=1}^{I} \alpha_{j,i} \mbgamma_{j,i} \mbbeta_{j,i}^\tp$ with $I = \min(n_\mathrm{p}, \Nrx)$, be the \ac{SVD} of the matrix $ \uvect_{\Nrx, n_\mathrm{p}} \left( \diag \left( (\mathbf{I} + \tfrac{1}{\sigma_n^2} \mbS_j \otimes \mbT_j )\inv \right) \right)$ for all $j \in \mathcal{J}$.
Then, it holds for all $j \in \mathcal{J}$,
\begin{align}
    & (\mathbf{I} + \tfrac{1}{\sigma_n^2} \mbS_j \otimes \mbT_j )\inv \nonumber \\
    & =\sum_{i=1}^I \diag(\alpha_{j,i} \mbbeta_{j,i} \otimes \mbgamma_{j,i}) = \sum_{i=1}^I \mbD_{\mbbeta_{j,i}} \otimes \mbD_{\mbgamma_{j,i}}, \label{eq:kpsvd_thm}
\end{align}
which are real-valued Kronecker product \acp{SVD}, see \cite{Lo00}, with the integer factorization of $n_\mathrm{p}$ and $\Nrx$.
Accordingly, the diagonal matrices $\mbD_{\mbbeta_{j,i}}\in \mathbb{R}^{n_\mathrm{p} \times n_\mathrm{p}}$ and $\mbD_{\mbgamma_{j,i}} \in \mathbb{R}^{\Nrx \times \Nrx}$ for all $j \in \mathcal{J}$, are defined as $\mbD_{\mbbeta_{j,i}} = \alpha_{j,i} \diag(\mbbeta_{j,i})$ and $\mbD_{\mbgamma_{j,i}} = \diag(\mbgamma_{j,i})$.
Note that the specific dimensions (the integer factorization of $n_\mathrm{p}$ and $\Nrx$) are required to proceed with simplifications involving Kronecker products.
Incorporating \eqref{eq:svd_inv_diag} and \eqref{eq:kpsvd_thm} into \eqref{eq:eldev_thm}, we obtain
\begin{align}
    &\sum_{j=1}^J \tr \Bigg( (\mbV_j \otimes \mbW_j) \left(\sum_{i=1}^I \mbD_{\mbbeta_{j,i}} \otimes \mbD_{\mbgamma_{j,i}}\right) (\mbV_j \otimes \mbW_j)^\herm \nonumber \\
    &\qquad \qquad \quad \times \left( \tfrac{1}{\sigma_n^2} (\mbP_t \mbC_{\mbdelta_j}^{\mathrm{tx}} \mathbf{e}_c \mathbf{e}_r^\tp)  \otimes \mbC_{\mbdelta_j}^{\mathrm{rx}} \right) \Bigg),   
\end{align}
which can be further reformulated by exploiting \eqref{eq:svd_covrx}, basic Kronecker product and trace rules to
\begin{align}
    &\sum_{j=1}^J \tr \Bigg( \left(\sum_{i=1}^I \mbD_{\mbbeta_{j,i}} \otimes \mbD_{\mbgamma_{j,i}}\right)  \nonumber \\
    &\qquad \qquad \quad \times  \left( \tfrac{1}{\sigma_n^2} (\mbV_j^\herm \mbP_t \mbC_{\mbdelta_j}^{\mathrm{tx}} \mathbf{e}_c \mathbf{e}_r^\tp \mbV_j)  \otimes \mbT_j \right) \Bigg) \nonumber \\
    &=\sum_{j=1}^J \tr\left( \sum_{i=1}^I \left( \tfrac{1}{\sigma_n^2} (\mbD_{\mbbeta_{j,i}}\mbV_j^\herm \mbP_t \mbC_{\mbdelta_j}^{\mathrm{tx}} \mathbf{e}_c \mathbf{e}_r^\tp \mbV_j)  \otimes (\mbD_{\mbgamma_{j,i}} \mbT_j) \right) \right)\\
    &=\sum_{j=1}^J  \sum_{i=1}^I \tr \left( \tfrac{1}{\sigma_n^2} \mbD_{\mbbeta_{j,i}}\mbV_j^\herm \mbP_t \mbC_{\mbdelta_j}^{\mathrm{tx}} \mathbf{e}_c \mathbf{e}_r^\tp \mbV_j \right) \tr (\mbD_{\mbgamma_{j,i}} \mbT_j)\\
    &=\sum_{j=1}^J  \sum_{i=1}^I \tfrac{1}{\sigma_n^2} \mathbf{e}_r^\tp \mbV_j \mbD_{\mbbeta_{j,i}}\mbV_j^\herm \mbP_t \mbC_{\mbdelta_j}^{\mathrm{tx}} \mathbf{e}_c \tr (\mbD_{\mbgamma_{j,i}} \mbT_j).
\end{align}
Accordingly, it holds:
\begin{align}
    &\dfrac{\partial}{\partial\mbP_t^*} \left(\sum_{j=1}^J \log\det \left( \mathbf{I} + \tfrac{1}{\sigma_n^2} (\mbP_t \mbC_{\mbdelta_j}^{\mathrm{tx}} \mbP_t^\herm)  \otimes \mbC_{\mbdelta_j}^{\mathrm{rx}} \right) \right) \nonumber \\
    &= \sum_{j=1}^J  \sum_{i=1}^I \tfrac{1}{\sigma_n^2} \mbV_j \mbD_{\mbbeta_{j,i}}\mbV_j^\herm \mbP_t \mbC_{\mbdelta_j}^{\mathrm{tx}} \tr (\mbD_{\mbgamma_{j,i}} \mbT_j).
\end{align}
It remains to set the derivative of the dual function of the optimization problem \eqref{eq:sumcmi_maximization} with respect to $\mbP_t^\star$ to zero to observe the first-order necessary \ac{KKT} condition provided in~\eqref{eq:sum_cmi_mumimo}, cf., e.g., \cite{BoVa04}.

\subsection{Proof of \Cref{thm:lower_bound}}
\label{sec:proof_lower_bound}

Given the reformulated sum \ac{CMI} expression from \eqref{eq:reform_scmi} and by using \eqref{eq:svd_covtx} and \eqref{eq:svd_covrx}, the sum \ac{CMI} can be expressed as
\begin{align}
    &\sum_{j=1}^J \log\det \left( \mathbf{I} + \tfrac{1}{\sigma_n^2} \mbS_j \otimes \mbT_j \right) \nonumber \\
    &=\sum_{j=1}^J \log\det \left( \mathbf{I} + \tfrac{1}{\sigma_n^2} \tr(\mbT_j) \mbS_j \otimes \tfrac{1}{\tr(\mbT_j)}\mbT_j \right). \label{eq:reform_scmi_tr}
\end{align}
By further defining for all $j \in \mathcal{J}$, $\tfrac{1}{\tr(\mbT_j)}\mbT_j = \sum_{i=1}^{\Nrx} t_{j,i} \diag(\mathbf{e}_i)$ with $\sum_{i=1}^{\Nrx} t_{j,i} = 1$, we can rewrite the expression provided in \eqref{eq:reform_scmi_tr} and lower bound it using the concavity of the log-determinant expression (see \cite[Th. 7.6.6]{HoJo85}) as
\begin{align}
    &\sum_{j=1}^J \log\det \left( \mathbf{I} + \tfrac{1}{\sigma_n^2} \tr(\mbT_j) \mbS_j \otimes \left(  \sum_i^{\Nrx} t_{j,i} \diag(\mathbf{e}_i) \right) \right) \nonumber \\
    &= \sum_{j=1}^J \log\det \sum_i^{\Nrx} t_{j,i} \left( \mathbf{I} + \tfrac{1}{\sigma_n^2} \tr(\mbT_j) \mbS_j \otimes  \diag(\mathbf{e}_i) \right) \label{eq:lb_ineq_left}\\
    &\geq \sum_{j=1}^J \sum_{i=1}^{\Nrx} t_{j,i} \log\det\left( \mathbf{I} + \tfrac{1}{\sigma_n^2} \tr(\mbT_j) \mbS_j \otimes  \diag(\mathbf{e}_i) \right) \label{eq:lb_ineq_right}\\
    &= \sum_{j=1}^J \sum_{i=1}^{\Nrx} t_{j,i} \log\det\left( \mathbf{I} + \tfrac{1}{\sigma_n^2} \tr(\mbT_j) \mbS_j \right) \\
    &= \sum_{j=1}^J \log\det\left( \mathbf{I} + \tfrac{1}{\sigma_n^2} \tr(\mbT_j) \mbS_j \right). \label{eq:lowerbound_end}
\end{align}
Since $\tr(\mbC_{\mbdelta_j}^{\mathrm{rx}}) = \tr(\mbT_j)$, see \eqref{eq:svd_covrx}, and by incorporating \eqref{eq:svd_covtx} in \eqref{eq:lower_bound}, and using $\mathbf{I} = \mbV_j \mbV_j^\herm $, yields \eqref{eq:lowerbound_end}, finishing the proof.

\subsection{Proof of \Cref{prop:lower_bound_equal}}
\label{sec:proof_lower_bound_equal}

Observe that the inequality in \eqref{eq:lb_ineq_right} becomes an equality if for all $j \in \mathcal{J}$ and any ${i^\prime}, t_{j,{i^\prime}} = 1 \wedge t_{j,i} = 0$ for all $i \neq {i^\prime}$, i.e., $\rk(\mbC_{\mbdelta_j}^{\mathrm{rx}}) = 1$ for all $j \in \mathcal{J}$.
In any other case, if there is at least one \ac{MT} $j$ for which $\rk(\mbC_{\mbdelta_j}^{\mathrm{rx}}) > 1$ holds with accordingly $t_{j,i}\in (0,1)$ for all $i$, and since then
\begin{equation} \label{eq:ineq_cond_matrices}
        \mathbf{I} + \tfrac{1}{\sigma_n^2} \tr(\mbT_j) \mbS_j \otimes  \diag(\mathbf{e}_i) \neq \mathbf{I} + \tfrac{1}{\sigma_n^2} \tr(\mbT_j) \mbS_j \otimes  \diag(\mathbf{e}_{i^\prime})
\end{equation}
due to the structure of the resulting matrices for $i \neq i^\prime$, the condition for equality between \eqref{eq:lb_ineq_left} and \eqref{eq:lb_ineq_right} according to \cite[Th. 7.6.6]{HoJo85} is not met (equality in \eqref{eq:ineq_cond_matrices} is required).

\subsection{Proof of \Cref{prop:lower_bound_largestgap}}
\label{sec:proof_lower_bound_largestgap}

The lower bound expression from \eqref{eq:lower_bound} only depends on the trace of each \acp{MT} receive-side covariance matrix $\mbC_{\mbdelta_j}^{\mathrm{rx}}$ for all $j \in \mathcal{J}$.
We outline next that the sum \ac{CMI} attains its largest value and consequently exhibits the largest gap to the lower bound for the case of no spatial correlation for all $j \in \mathcal{J}$.
Thus, given the reformulated sum \ac{CMI} from \eqref{eq:reform_scmi_tr}, we have to show
\begin{align}
    &\sum_{j=1}^J \log\det \left( \mathbf{I} + \tfrac{1}{\sigma_n^2} \tau_j \mbS_j \otimes \tfrac{1}{\Nrx} \mathbf{I} \right) \nonumber \\
    & \stackrel{!}> \sum_{j=1}^J \log\det \left( \mathbf{I} + \tfrac{1}{\sigma_n^2} \tau_j \mbS_j \otimes \tfrac{1}{\Nrx} (\mathbf{I} + \diag(\mbkappa_j)) \right),
\end{align}
where for all $j \in \mathcal{J}$, the elements of $\mbkappa_j$ satisfy $\sum_{\ell=1}^{\Nrx} \kappa_{j,\ell} = 0$, with $\kappa_{j,\ell} + 1 \geq0$, and $\mbkappa_j \neq \mathbf{0}$ (in this way, for all $j \in \mathcal{J}$, the eigenvalues of any valid receive-side covariance matrix, except the weighted identity, are parametrized).
Since the involved matrices are diagonal, we have
\begin{align}
    &\sum_{j=1}^J \log \prod_{i=1}^{n_\mathrm{p}} (1+\tfrac{\tau_j}{\sigma_n^2 \Nrx} s_{j,i})^{\Nrx} \nonumber \\
    & \stackrel{!}> \sum_{j=1}^J \log  \prod_{i=1}^{n_\mathrm{p}} \prod_{\ell=1}^{\Nrx}(1+\tfrac{\tau_j}{\sigma_n^2 \Nrx} s_{j,i} (1+\kappa_{j,\ell})),
\end{align}
which can be further rewritten to
\begin{align}
    &\sum_{j=1}^J \sum_{i=1}^{n_\mathrm{p}} \log (1+\tfrac{\tau_j}{\sigma_n^2 \Nrx} s_{j,i})^{\Nrx} \nonumber \\
    & \stackrel{!}> \sum_{j=1}^J  \sum_{i=1}^{n_\mathrm{p}} \sum_{\ell=1}^{\Nrx} \log(1+\tfrac{\tau_j}{\sigma_n^2 \Nrx} s_{j,i} (1+\kappa_{j,\ell})).
\end{align}
For a fixed $j$ and $i$, we thus have to show 
\begin{align}
    \Nrx \log (1+\tfrac{\tau_j}{\sigma_n^2 \Nrx} s_{j,i}) \stackrel{!}> \sum_{\ell=1}^{\Nrx} \log(1+\tfrac{\tau_j}{\sigma_n^2 \Nrx} s_{j,i} (1+\kappa_{j,\ell})),
\end{align}
which can be reformulated by dividing both sides by $\Nrx$ as
\begin{align}
    &\log \left(\sum_{\ell=1}^{\Nrx}\tfrac{1}{\Nrx}(1+\tfrac{\tau_j}{\sigma_n^2 \Nrx} s_{j,i})\right) \label{eq:lbmaxgap_before} \nonumber \\
    & \stackrel{!}> \sum_{\ell=1}^{\Nrx}  \tfrac{1}{\Nrx} \log(1+\tfrac{\tau_j}{\sigma_n^2 \Nrx} s_{j,i} (1+\kappa_{j,\ell})).
\end{align}
Observing that the expression provided in \eqref{eq:lbmaxgap_end} is equal to \eqref{eq:lbmaxgap_before} and the following inequality is true due to Jensen's inequality
\begin{align}
    &\log \left(\sum_{\ell=1}^{\Nrx}\tfrac{1}{\Nrx}(1+\tfrac{\tau_j}{\sigma_n^2 \Nrx} s_{j,i} (1+\kappa_{j,\ell}))\right) \label{eq:lbmaxgap_end} \nonumber \\
    & > \sum_{\ell=1}^{\Nrx}  \tfrac{1}{\Nrx} \log(1+\tfrac{\tau_j}{\sigma_n^2 \Nrx} s_{j,i} (1+\kappa_{j,\ell})),
\end{align}
and is strict since $\mbkappa_j \neq \mathbf{0}$, the only condition for which equality can hold is not satisfied, i.e. $1+\tfrac{\tau_j}{\sigma_n^2 \Nrx} s_{j,i} (1+\kappa_{j,\ell}) \neq 1+\tfrac{\tau_j}{\sigma_n^2 \Nrx} s_{j,i} (1+\kappa_{j,{\ell^\prime}})$ for $\ell \neq \ell^{\prime}$. 

\balance
\bibliographystyle{IEEEtran}
\bibliography{IEEEabrv,biblio}
\end{document}